\documentclass[a4paper,UKenglish]{lipics-v2018}
 
\usepackage{microtype}

\usepackage{graphicx}
\usepackage{float}
\usepackage{enumerate}
\newcommand{\lp}{\mathsf{left}}
\newcommand{\rp}{\mathsf{right}}
\newcommand{\lnk}{\mathsf{link}}
\newcommand{\ct}{\mathsf{cat}}
\renewcommand{\P}{\mathcal{P}}

\title{Improved bounds for multipass pairing heaps and path-balanced binary search trees}
\titlerunning{Improved bounds for multipass pairing heaps and path-balanced binary search trees} 

\author{Dani Dorfman}{Blavatnik School of Computer Science, Tel Aviv University, Israel}{dannatand@mail.tau.ac.il}{}{}

\author{Haim Kaplan}{Blavatnik School of Computer Science, Tel Aviv University, Israel}{haimk@post.tau.ac.il}{}{Research supported by The Israeli Centers of Research Excellence (I-CORE) program (Center No.\ 4/11), Israel Science Foundation grant no.\ 1841-14.}

\author{L\'{a}szl\'{o} Kozma}{Eindhoven University of Technology, The Netherlands}{lkozma@gmail.com}{}{Supported by ERC grant no.\ 617951. Work done while at Tel Aviv University, research supported by
The Israeli Centers of Research Excellence (I-CORE) program (Center
No.\ 4/11).}

\author{Seth Pettie}{University of Michigan}{pettie@umich.edu}{}{Supported by NSF grants CCF-1514383 and CCF-1637546.}

\author{Uri Zwick}{Blavatnik School of Computer Science, Tel Aviv University, Israel}{zwick@tau.ac.il}{}{Research supported by BSF grant no.\ 2012338 and by
The Israeli Centers of Research Excellence (I-CORE) program (Center
No.\ 4/11).}

\authorrunning{D.\,Dorfman, H.\,Kaplan, L.\,Kozma, S.\,Pettie and U.\,Zwick} 

\Copyright{Dani Dorfman, Haim Kaplan, L\'{a}szl\'{o} Kozma, Seth Pettie, Uri Zwick}

\date{}

\EventEditors{Yossi Azar, Hannah Bast, and Grzegorz Herman}
\EventNoEds{3}
\EventLongTitle{26th Annual European Symposium on Algorithms (ESA 2018)}
\EventShortTitle{ESA 2018}
\EventAcronym{ESA}
\EventYear{2018}
\EventDate{August 20--22, 2018}
\EventLocation{Helsinki, Finland}
\EventLogo{}
\SeriesVolume{112}
\ArticleNo{24} 

\nolinenumbers 
\hideLIPIcs  

\begin{document}

\maketitle

\begin{abstract}
We revisit \emph{multipass} pairing heaps and \emph{path-balanced} binary search trees (BSTs), two classical algorithms for data structure maintenance. The pairing heap is a simple and efficient ``self-adjusting'' heap, introduced in 1986 by Fredman, Sedgewick, Sleator, and Tarjan. In the multipass variant (one of the original pairing heap variants described by Fredman et al.) the minimum item is extracted via repeated \emph{pairing rounds} in which neighboring siblings are linked.   

Path-balanced BSTs, proposed by Sleator (cf.\ Subramanian, 1996), are a natural alternative to Splay trees (Sleator and Tarjan, 1983). In a path-balanced BST, whenever an item is accessed, the search path leading to that item is re-arranged into a balanced tree.  

Despite their simplicity, both algorithms turned out to be difficult to analyse. Fredman et al.\ showed that operations in multipass pairing heaps take amortized $O(\log{n} \cdot \log\log{n} / \log\log\log{n})$ time. For searching in path-balanced BSTs, 
Balasubramanian and Raman showed in 1995 the same amortized time bound of $O(\log{n} \cdot \log\log{n} / \log\log\log{n})$, using a different argument.

In this paper we show an explicit connection between the two algorithms and improve the two bounds to $O\left(\log{n} \cdot 2^{\log^{\ast}{n}} \cdot \log^{\ast}{n}\right)$, respectively $O\left(\log{n} \cdot 2^{\log^{\ast}{n}} \cdot (\log^{\ast}{n})^2 \right)$, where $\log^{\ast}(\cdot)$ denotes the very slowly growing iterated logarithm function. These are the first improvements in more than three, resp.\ two decades, approaching in both cases the information-theoretic lower bound of $\Omega(\log{n})$.  

 \end{abstract}

\section{Introduction}
Binary search trees (BSTs) and heaps are the canonical comparison-based implementations of the well-known \emph{dictionary} and \emph{priority queue} data types. 

In a balanced \textbf{BST} all standard dictionary operations (\emph{insert}, \emph{delete}, \emph{search}) take $O(\log{n})$ time, where $n$ is the size of the dictionary. Early research has mostly focused on structures that are kept (approximately) balanced throughout their usage. (AVL-, red-black-trees, and randomized treaps are important examples, see e.g., \cite[\S\,6.2.2]{Knuth3}). These data structures re-balance themselves when necessary, guided by auxiliary data stored in every node.  

By contrast, Splay trees (Sleator, Tarjan, 1983~\cite{ST85}) achieve $O(\log{n})$ amortized time per operation without any explicit balancing strategy and with no bookkeeping whatsoever. Instead, Splay trees re-adjust the search path \emph{after every access}, in a way that depends only on the shape of the search path, ignoring the global structure of the tree. Besides the $O(\log{n})$ amortized time, Splay trees are known to satisfy stronger, adaptive properties (see~\cite{in_pursuit, landscape} for surveys). They are, in fact, conjectured to be optimal on every sequence of operations (up to a constant factor); this is the famous ``dynamic optimality conjecture''~\cite{ST85}. Splay trees and data structures of a similar flavor (i.e., local restructuring, adaptivity, no auxiliary data) are called ``self-adjusting''.

The efficiency of Splay trees is intriguing and counter-intuitive. They re-arrange the search path by a sequence of double rotations (``zig-zig'' and ``zig-zag''), bringing the accessed item to the root. It is not hard to see that this transformation results in ``approximate depth-halving'' for the nodes on the search path; the connection between this depth-halving and the overall efficiency of Splay trees is, however, far from obvious.  

An arguably more natural approach for BST re-adjustment would be to turn the search path, after every search, into a balanced tree.\footnote{The restriction to touch only the search path is natural, as the cost of doing this is proportional to the \emph{search cost}. (A BST can be changed into any other BST with a linear number of rotations~\cite{STT88}.)}  
This strategy combines the idea of self-adjusting trees with the more familiar idea of balancedness. 
Indeed, this algorithm was proposed early on by Sleator (see e.g., \cite{Subramanian96, pathbalance}). We refer to BSTs maintained in this way as \emph{path-balanced} BSTs (see Figure~\ref{fig_pb}). 

Path-balanced BSTs turn out to be surprisingly difficult to 
analyse. In 1995, Balasubramanian and Raman~\cite{pathbalance} 
showed the upper bound of 
$O(\log{n} \cdot \log\log{n} / \log\log\log{n})$ on the cost of 
operations in path-balanced BSTs. This bound has not been improved 
since. Thus, path-balanced BSTs are not known to match the 
$O(\log{n})$ amortized cost (let alone the stronger adaptive 
properties) of Splay. This is surprising, because broad classes of 
BSTs are known to match several guarantees of Splay trees~\cite{Subramanian96, ESA15}, path-balanced BSTs, however, fall outside these classes.\footnote{Intuitively, path-balance is different, and more difficult to analyse than Splay, because it may increase the depth of a node by an additive $O(\log{n})$, whereas Splay may increase the depth of a node by at most $2$. In a precise sense, path-balance is not a \emph{local} transformation (see~\cite{ESA15}).} Without evidence to the contrary, one may even conjecture path-balanced BSTs to achieve dynamic optimality; yet our current upper bounds do not even match those of a \emph{static} balanced tree. This points to a large gap in our understanding of a natural heuristic in the fundamental BST model.

In this paper we show that the amortized time of an access\footnote{We only focus on successful search operations (i.e., accesses). The results can be extended to other operations at the cost of technicalities. For simplicity, we assume that the keys in the tree are unique.} in a path-balanced BST is $O \left( \log{n} \cdot \left( \log^{\ast}{n} \right)^2 \cdot 2^{\log^{\ast}{n}} \right)$. The result, probably not tight, comes close to the information-theoretic lower bound of $\Omega(\log{n})$. Closing the gap remains a challenging open problem. \\ 

\begin{figure}
	\begin{center}
\includegraphics[width=0.6\textwidth]{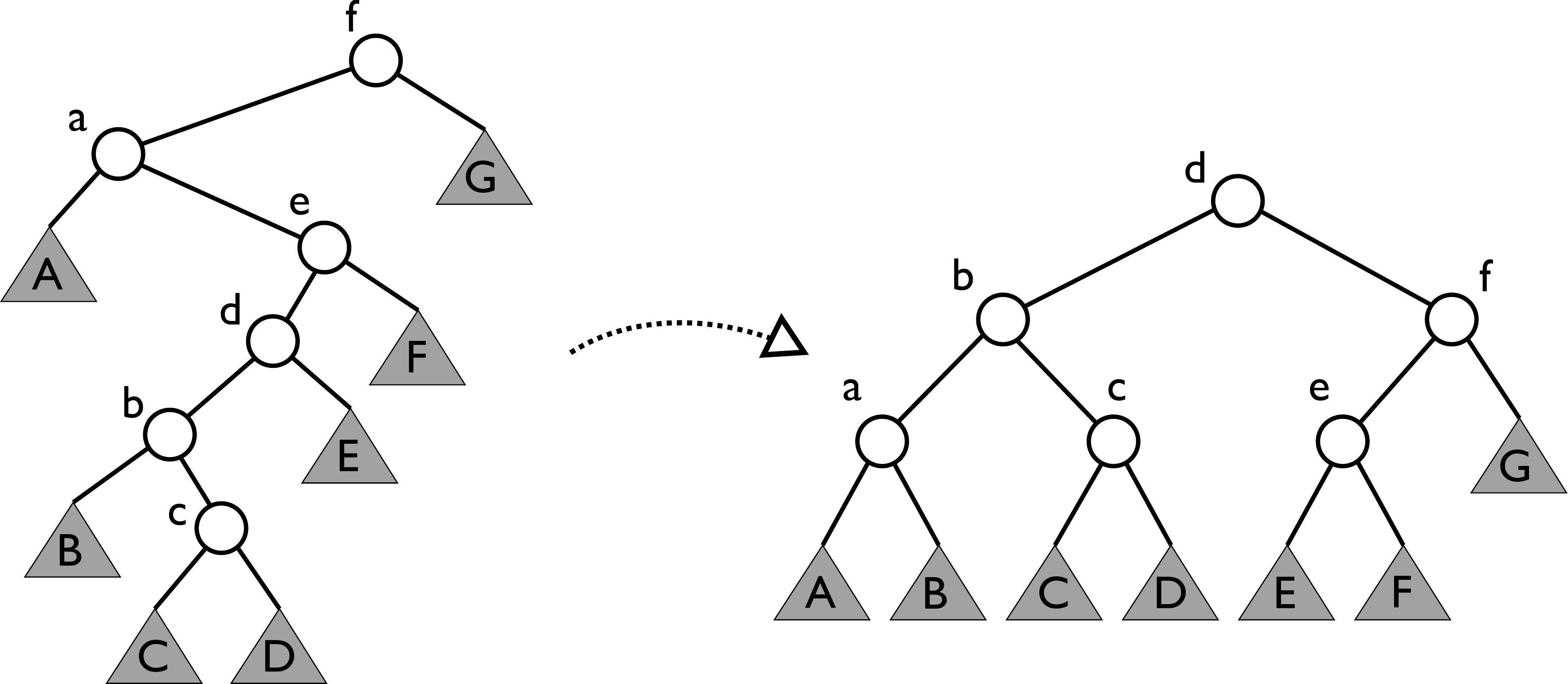}
	\end{center}
\caption{Access in a path-balanced BST. Search path $(f,a,e,d,b,c)$ from root $f$ to accessed item $c$ is re-arranged into a balanced tree with subtrees (denoted by capital letters) re-attached.
\label{fig_pb}}
\end{figure}

\textbf{Priority queues} support the operations \emph{insert}, \emph{delete-min}, and possibly \emph{meld}, \emph{decrease-key} and others. Pairing heaps, a popular priority queue implementation, were proposed in the 1980s by Fredman, Sedgewick, Sleator, and Tarjan~\cite{pairing} as a simpler, self-adjusting alternative to Fibonacci heaps~\cite{Fibonacci}. Pairing heaps maintain a multi-ary tree whose nodes (each with an associated key) are in heap order. Similarly to Splay trees, pairing heaps only perform key-comparisons and simple local transformations on the underlying tree, with no auxiliary data stored. Fredman et al.\ showed that in the standard pairing heap all priority queue operations take $O(\log{n})$ time. They also proposed a number of variants, including the particularly natural \emph{multipass pairing heap}. In multipass pairing heaps, the crucial \emph{delete-min} operation is implemented as follows. After the root of the heap  (i.e., the minimum) is deleted, repeated pairing rounds are performed on the new top-level roots, reducing their number until a single root remains. In each pairing round, neighboring pairs of nodes are \emph{linked}. Linking two nodes makes the one with the larger key the \emph{leftmost} child of the other (Figure~\ref{fig_mp}). 

Pairing heaps perform well in practice~\cite{Stasko,Moret,Tarjan14}. However, Fredman~\cite{FredmanLB} showed that all of their standard variants (including the multipass described above) fall short of matching the theoretical guarantees of Fibonacci heaps (in particular, assuming $O(\log{n})$ cost for delete-min, the average cost of \emph{decrease-key} may be $\Omega(\log\log{n})$, in contrast to the $O(1)$ guarantee for Fibonacci heaps). The exact complexity of the standard pairing heap on sequences of intermixed delete-min, insert, and decrease-key operations remains an intriguing open problem, with significant progress through the years (see e.g., \cite{IaconoUB, Pettie}). However, for the multipass variant, even the basic question of whether deleting the minimum takes $O(\log{n})$ amortized time remains open, the best upper bound to date being the $O(\log{n} \cdot \log\log{n} / \log\log\log{n})$ originally shown by Fredman et al. 
Similarly to the case of path-balanced BSTs, we have thus a basic combinatorial transformation on trees, whose complexity is not well-understood. 

\begin{figure}
	\begin{center}
\includegraphics[width=0.7\textwidth]{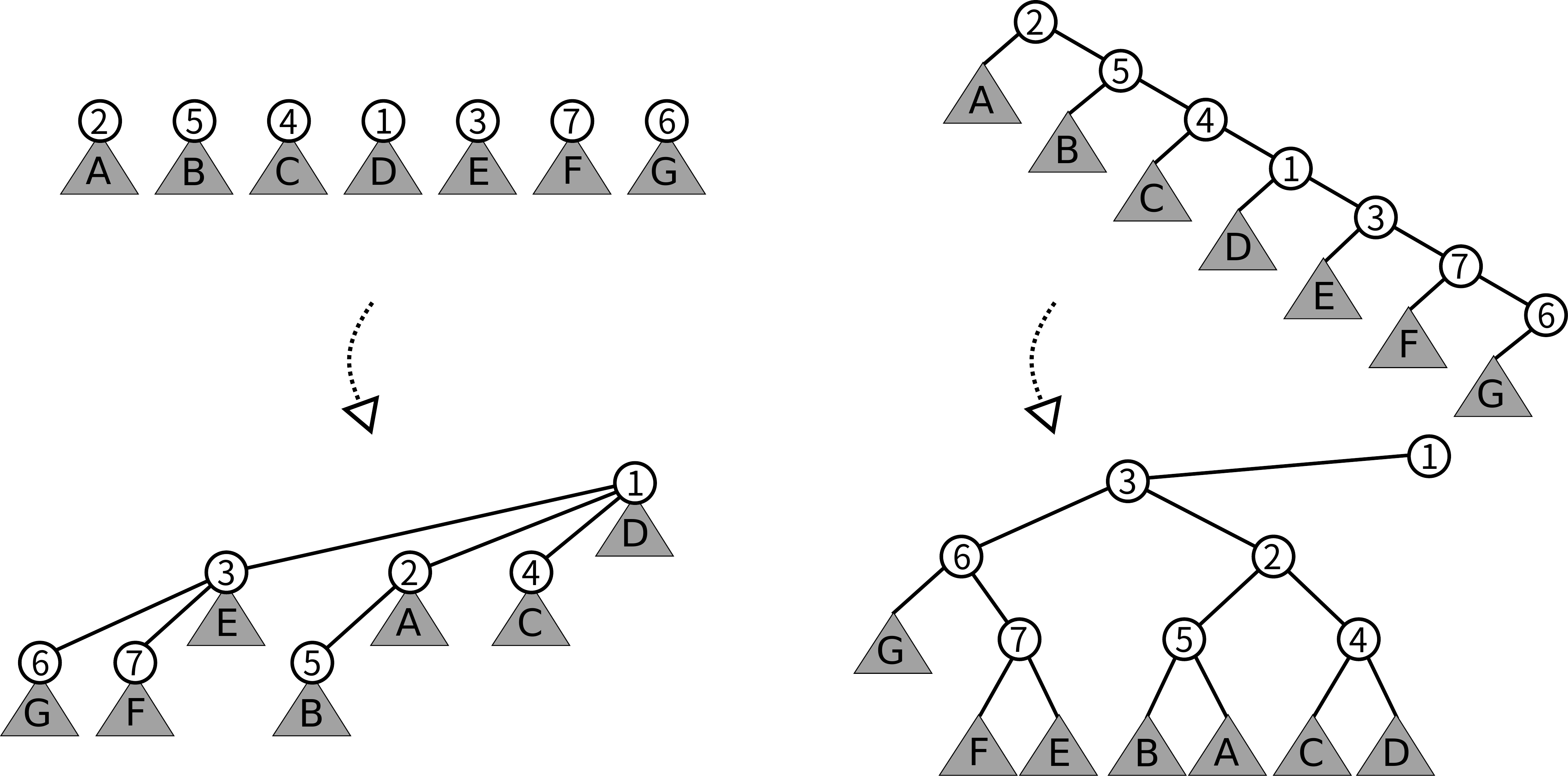}
	\end{center}
\caption{Delete-min in a multipass pairing heap. (above) state after deleting the root, with list of siblings; (below) state after three pairing rounds, with links $(2,5)$, $(4,1)$, $(3,7)$, $(2,1)$, $(3,6)$, $(1,3)$. (left) multi-ary view; (right) binary view. Numbers denote keys, capital letters denote subtrees.
\label{fig_mp}}
\end{figure}

In this paper we show that in multipass pairing heaps delete-min\footnote{To keep the presentation simpler, we only focus on \emph{delete-min} operations, omitting the extension of the result to other operations.} takes amortized time $O \left( \log{n} \cdot \log^{\ast}{n}  \cdot 2^{\log^{\ast}{n}} \right)$, the first improvement since the original paper of Fredman et al.
The improvement is, from a practical perspective, not significant. Nonetheless, it reduces the gap to the theoretical optimum from $(\approx \log^{(2)}{n})$ to less than $\log^{(k)}{n}$ for any fixed $k$. 

The reader may notice that the old bounds for multipass pairing heaps and path-balanced BSTs are the same. The two data structures are, indeed, quite similar: if one views multipass pairing heaps as binary trees (see e.g., \cite[\S\,2.3.2]{Knuth}), the multipass re-adjustement is equivalent to balancing the right-spine of a binary tree.\footnote{We note that the previous analysis of path-balanced BSTs~\cite{pathbalance} did not use this correspondence. By connecting the two data structures, we also simplify (to some extent) the proof of~\cite{pathbalance}.} The multipass analysis, however, does not immediately transfer to path-balanced BSTs; the fact that the BST search path may be arbitrary (not necessarily right-leaning) complicates the argument for path-balanced BSTs. 

Our analysis of multipass pairing heaps (\S\,\ref{sec:mp}) is based on a new, fine-grained scaling of the sum-of-logs potential function used by Sleator and Tarjan in the analysis of Splay trees, and by Fredman et al.\ in the analysis of pairing heaps. At a high level, we argue that certain link operations are information-theoretically efficient, and that such links happen sufficiently often. The subsequent, rather intricate analysis notwithstanding, we believe that the ideas of the proof may have further applications in the analysis of data structures.

In \S\,\ref{sec:pb} we show our result for path-balanced BSTs. Informally, we decompose the path-balancing operation into several stages, each of which resembles the multipass transformation, allowing us to adapt and reuse the result of \S\,\ref{sec:mp}.

\section{Multipass pairing heaps}
\label{sec:mp}
A \emph{pairing heap} is a multi-ary heap, storing a key in each node, with the regular (min)heap-condition: the key of a node is smaller than the keys of its children. Priority queue operations are implemented using the unit-cost \emph{linking} step. Given nodes $x,y$, $\lnk(x,y)$ ``hangs'' the node with the larger key as the \emph{leftmost} child of the other. The operations \emph{insert}, \emph{meld}, and \emph{decrease-key} can be implemented in a straightforward way using a single \emph{link} (we refer to~\cite{pairing} for details). The only nontrivial operation is \emph{delete-min}. Here, after deleting the root, we are left with a number of top-level nodes, which we combine into a single tree via a sequence of \emph{links}. In multipass pairing heaps we achieve this by performing repeated \emph{pairing rounds}, until a single top-level node remains (i.e., the new root of the heap). A single pairing round is as follows. Let $x_1, \dots, x_\ell$ be the top-level nodes, ordered left-to-right, before the round. For all $1\leq i \leq \lfloor \ell/2 \rfloor$ we perform $\lnk(x_{2i-1},x_{2i})$. Observe that if $\ell$ is odd, then the rightmost node is unaffected in the current round. The number of rounds is $\lceil \log(k) \rceil$, where $k$ is the number of children of the (deleted) root.\footnote{The function $\log(\cdot)$ is base $2$ everywhere, the base $e$ logarithm is written as $\ln(\cdot)$.} (See Figure~\ref{fig_mp}.)

We now analyse delete-min operations implemented by multipass pairing heaps. Let $k$ be the number of children of the deleted root, defined to be the real cost of the operation (observe that the number of links is exactly $k-1$). Let $n$ be the size of the heap before the operation. We use the binary tree view of multi-ary heaps, where the \emph{leftmost child} and \emph{next sibling} pointers are interpreted as \emph{left child} and \emph{right child}. 
A single link operation is shown in Figure~\ref{fig1}.  Let $a$, $b$, $c$ denote the sizes of subtrees $A$, $B$, and $C$, respectively. 

\begin{figure}[H]
    \begin{center}
\includegraphics[width=0.95\textwidth]{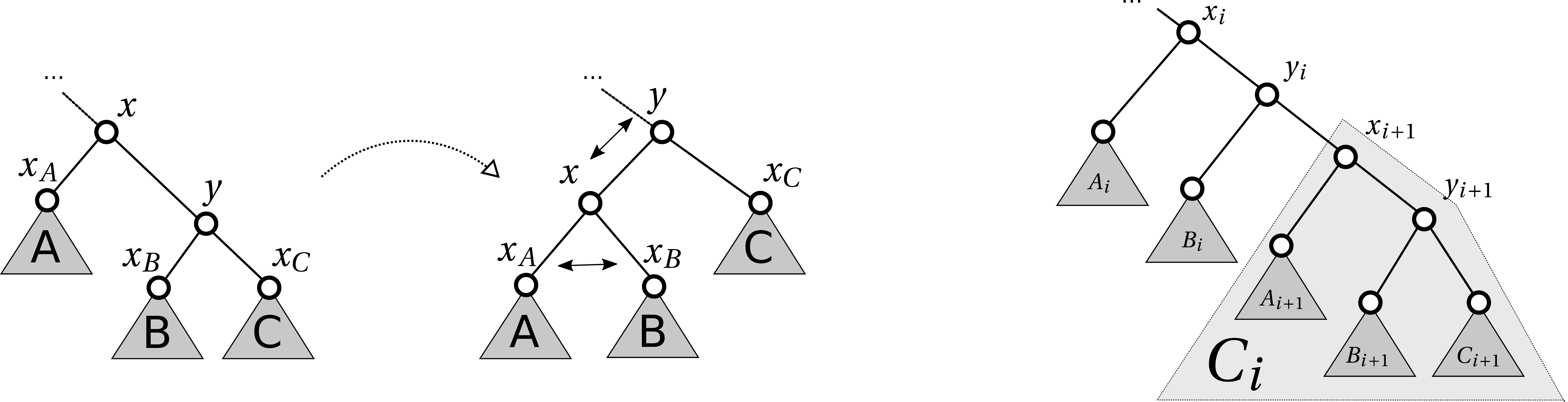}
    \end{center}
\caption{Left: $\lnk(x,y)$ in binary tree view. Dots (...) indicate the sequence of nodes that have already been linked in the current round, subtree $C$ contains the yet-to-be-linked nodes. Arrows indicate possible switching depending on the outcome of the comparison between $x$ and $y$. The roots of $A$,$B$, and $C$ are denoted $x_A$, $x_B$, and $x_C$. Right: $i$-th link in a round (between $x_i$ and $y_i$). The subtree rooted at the right child of $y_i$ is denoted $C_i$; observe that $C_i$ contains $C_{i+1}$. 
\label{fig1}}
\end{figure}

 We define a potential function that refines the Sleator-Tarjan ``sum-of-logs'' potential~\cite{ST85}. 
Let $\Phi = \sum_{x \in T}{\phi(x)}$, over all nodes $x$ of the heap $T$, where $$\phi(x) = \frac{H(x)}{\log^2{(2+H(x))}},~~\mbox{and}~~H(x) = \log{\left( \frac{ s(p(x))}{s(x)} \right)},$$ where $s(x)$ denotes the size of the subtree rooted at $x$, and $p(x)$ is the parent of $x$.\footnote{Using $\phi(x) = H(X)$ instead, would essentially recover the original ``sum-of-logs'' potential. Such an ``edge-based'' potential function was used earlier, e.g., in~\cite{GeorgakopoulosM04,Mehlhorn84}.} Note that both \emph{subtrees} and \emph{parents} are meant in the binary tree view.

For convenience, define the functions $$f(x) = \log{x}/\log^2{(2+\log{x})},~~\mbox{ and}~~g(x) = x/\log^2{(2+x)}.$$ With this notation, $f(x) = g(\log{(x)})$, and $\displaystyle\phi(x) = f\left(\frac{s(p(x))}{s(x)}\right)$.
Clearly, both $f(x)$ and $g(x)$ are positive, monotone increasing, and concave, for all $x\geq 1$, respectively, $x \geq 0$. 

By simple arithmetic, the increase in potential due to a single link (as in Figure~\ref{fig1}) is:
\begin{align}
\Delta \Phi \: = 
\begin{split}
f\left(\frac{a+b+1}{a}\right) + f\left(\frac{a+b+1}{b}\right) + f\left(\frac{a+b+c+2}{a+b+1}\right) + f\left(\frac{a+b+c+2}{c}\right) \label{eq:1} \\
- f\left(\frac{a+b+c+2}{a}\right) - f\left(\frac{a+b+c+2}{b+c+1}\right) - f\left(\frac{b+c+1}{b}\right) - f\left(\frac{b+c+1}{c}\right).
\end{split}
\end{align}

For a suitably large constant $\gamma$ (for concreteness let $\gamma=3000$), we consider the quantities $\gamma^2 a$, $\gamma b$, and $c$, i.e., the scaled sizes of the subtrees $A$, $B$, and $C$. We distinguish different kinds of links, depending on the ordering of the three quantities (breaking ties arbitrarily).
We first look at the cases when $\gamma^2 a$ or $\gamma b$ is the largest (called respectively type-(1) and type-(2) links), and show that the possible increase in potential due to such links is small. In particular, for type-(1) links, $\Delta \Phi$ is dominated by a term $f(a/c)$, and for type-(2) links the positive and negative contributions cancel out, leaving $\Delta \Phi = O(1)$. The proofs use standard (although somewhat delicate) analysis; we defer most of the calculations to Appendix~\ref{appa}.

\begin{lemma}[\ref{lem1proof}]\label{lema}
A type-(1) link \emph{($\gamma^2 a \geq \max{\{\gamma b, c \} }$)} increases the potential $\Phi$ by at most $2 \cdot g\bigl( \log{(a/c)} + O(1) \bigr)$, where   
the $O(1)$ term is a constant independent of $a$, $b$, $c$, $n$, and $k$.
\end{lemma}

\begin{lemma}[\ref{lem2proof}]\label{lemb}

A type-(2) link \emph{($\gamma b \geq \max{\{\gamma^2 a, c \} }$)} increases $\Phi$ by at most $O(1)$.

\end{lemma}

The case when $c$ is the greatest of the three quantities (called type-(3) link) is the most favorable. Here, the potential of $x_A$, $x_B$ before the linking is (roughly) the logarithm of $s(x_C)$ (very large) divided by $s(x_A),s(x_B)$; after the linking, the potential becomes (essentially) the logarithm of the ratio between $s(x_A)$ and $s(x_B)$ (much smaller), resulting in a significant saving in potential. We use this saving to ``pay'' for the operations. First we make the following, easier claim. 

\begin{lemma}[\ref{lemcproof}]\label{lemc}
A type-(3) link \emph{($c \geq \max{\{\gamma^2 a, \gamma b \} }$)} can not increase $\Phi$.
\end{lemma}

It remains to balance the \emph{decrease} in potential due to type-(3) links and the \emph{increase} in potential due to all other links. First, we show that almost all links are type-(3).

\begin{lemma}\label{lem5}
There are at most $O(\log{n})$ type-(1) and type-(2) links within a pairing round.
\end{lemma}
\begin{proof}
Let $a_i$, $b_i$, $c_i$ denote the subtree-sizes corresponding to the $i$-th link \emph{from left to right}, see Figure~\ref{fig1}(right). Let the subsequences $a_{i_t}$, $b_{i_t}$, $c_{i_t}$, $t=1,\dots,m$ be the subtree-sizes corresponding to type-(1) and type-(2) links. Observe that $c_{i_1} \geq \cdots \geq c_{i_m}$. If the $i$-th link is of type-(1) or type-(2), then $c_{i-1} = 2 + a_i + b_i + c_i \geq (1+1/\gamma^2) \cdot c_{i}$, since in each of these cases $a_i \geq 1/\gamma^2 c_i$ or $b_i \geq 1/\gamma^2 c_i$. Since $c_{i_1} \leq n$, and $c_{i_m} \geq 1$ the claim follows. 
\end{proof}

\begin{lemma}\label{lemlevelpot}
All type-(1) and type-(2) links within a single pairing round increase the potential by at most $O(\log{n})$.
\end{lemma}

\begin{proof}
Look at a single round of pairing.
Let $a_{i_t}$, $b_{i_t}$, $c_{i_t}$ ($t=1, \dots, m$) be as in the proof of Lemma~\ref{lem5} and recall that $m = O(\log{n})$.
If the $i_t$-th link is type-(1), then by Lemma~\ref{lema}, the increase in potential is at most 
$2 \cdot g\bigl( \log{(a_{i_t}/c_{i_t})} + O(1) \bigr)$.

Otherwise, if the $i_t$-th link is type-(2), then by Lemma~\ref{lemb}, the increase in potential is at most $O(1)$, which we can write as $2 \cdot g(c')$, for a suitable constant $c'$.

Let $q_t$ denote $\log{(a_{i_t}/c_{i_t})} + O(1)$, or $c'$, corresponding to the $i_t$-th link (according to its type).
We have $\sum{q_i} \leq \alpha \cdot \log{n}$ (for a fixed constant $\alpha \geq 1$), since the sum of the $\log{(a_i/c_i)}$ terms telescopes, and the additive $O(1)$ (or $c'$) terms appear at most $m = O(\log{n})$ times.

The total increase in potential is at most $\Delta \Phi = 2 \cdot \sum_{t=1}^{m}{g(q_t)}$. 
By the concavity of $g(\cdot)$, $\Delta \Phi$ is maximized if all of the arguments of $g(\cdot)$ are equal. We thus obtain a bound on the total increase in potential in the pairing round.
$$\Delta \Phi  \leq  2 m \cdot g\left(\frac{\alpha \cdot \log{n}}{{m}}\right)
= \frac{2\alpha \log{n}}{\log^2{(2 + \alpha \cdot (\log{n})/{m})}}
= O\left( \log{n} \right).\quad \qedhere$$
\end{proof}

The last proof yields, in fact, the following stronger claim.

\begin{lemma}\label{lempot}
All type-(1) and type-(2) links within the last $(\log\log{n})$ pairing rounds increase the potential by at most $O(\log{n})$.
\end{lemma}

\begin{proof}
Observe that for $j<\log\log{n}$, the $j$-th to the last pairing round has at most $m \leq 2^j < \log{n}$ links. 
Thus, as in Lemma~\ref{lemlevelpot}, we obtain:
$$\Delta \Phi
\leq \frac{2\alpha \log{n}}{\log^2{(2 + \alpha \cdot (\log{n})/{m})}}\\
\leq  \frac{2\alpha \log{n}}{\log^2{(\alpha \cdot (\log{n})/{2^j})}}\\
=  \frac{2\alpha \log{n}}{((\log\log{n} + \log{\alpha})-j)^2}.$$

Note that the second inequality holds since $2^j < \log{n}$.
The sum of this expression over all $(\log{\log{n}})$ levels $j$ is $O(\log{n})$. (Using the fact that $\sum_k{1/k^2}$ converges to a constant.)
\end{proof}
Now we estimate more carefully the decrease in potential due to type-(3) links. Let $x_A$ and $x_B$ be nodes as denoted in Figure~\ref{fig1}. We want to express the potential-change in terms of $H_A = H(x_A)$ and $H_B = H(x_B)$ (before the link operation).
Recall that $H_A =\log \left(\frac{a+b+c+2}{a}\right)$ and $H_B = \log \left(\frac{b+c+1}{b}\right)$.

Among type-(3) links ($ c \geq \max{\{\gamma^2 a,\gamma b\}}$) we distinguish two subtypes: type-(3A) ($ \gamma^2 a \geq \gamma b$), and type-(3B) ($ \gamma b \geq \gamma^2 a$). We have the following two (symmetric) observations:

\begin{lemma}[\ref{app_lem7}]\label{cat-type5}
A type-(3A) link \emph{($ c \geq \gamma^2 a \geq \gamma b$)} decreases the potential by at least
$$
\Omega(1) \cdot \frac{H_A}{\log^2{(2+H_B)}} - O(1).
$$
\end{lemma}

It follows that for some constant $d_1$, if $H_A \geq d_1 \cdot \log^2{(2+H_B)}$, then $\Delta\Phi \leq -1$.

\begin{proof}
Let $x=c/a \geq \gamma^2$ and $y=a/b \geq 1/\gamma$. Then, recalling Equation~(\ref{eq:1}): 
\vspace{-0.1in}
\begin{align*}
-\Delta\Phi & =  f(1 + \frac{y}{1+xy}) + f(1 + 1/y + x) + f(1 + \frac{1}{xy}) + f(1 + xy)\\
& -f(1 + y) - f(1+\frac{1}{y}) - f(1+\frac{1}{x}+\frac{1}{xy}) - f(1 + \frac{xy}{y+1}) - O(1).
\end{align*}

We have $H_A = \log(1+x+1/y) + O(1)$, and $H_B = \log(1+xy)+O(1)$. Note, $H_B \geq \Omega(1) \cdot H_A$.

Collecting constant terms, we have:
$$-\Delta\Phi \geq  f(1 + xy) + f(1 + x) - f(1+y) - f\left(1 + \frac{xy}{y+1}\right) - O(1).$$

As $f(1 + x) - f(1 + \frac{xy}{y+1}) \geq 0$, we further simplify:
$-\Delta\Phi \geq  f(1 + xy)  - f(1+y) - O(1)$.

It is now sufficient to show: $$f(1 + xy)  - f(1+y) \geq \Omega(1) \cdot \frac{\log{(1+x+1/y)}}{\log^2{\left(2 + \log{(1+xy)}\right)}} - O(1).$$

(We defer the detailed calculations to \ref{app_lem7}.) \qedhere
\end{proof}

\begin{lemma}[\ref{app_lem8}]\label{cat-type6}
A type-(3B) link \emph{($ c \geq \gamma b \geq \gamma^2 a$)} decreases the potential by at least
$$
\Omega(1) \cdot \frac{H_B}{\log^2{(2+H_A)}} - O(1).
$$
\end{lemma}

It follows that for some constant $d_2$, if $H_B \geq d_2 \cdot \log^2{(2+H_A)}$, then $\Delta\Phi \leq -1$. 

\begin{corollary}\label{cor-categories}
There exists a constant $d$ \emph{($= \max(d_1,d_2)$)} such that all type-(3A) links with $H_A \geq d \cdot \log^2{(2+H_B)}$ and all type-(3B) links with $H_B \geq d \cdot \log^2{(2+H_A)}$ decrease the potential by at least $1$.
\end{corollary}
We now define the \emph{category} of a node with respect to its $H(\cdot)$ value. Intuitively, nodes of the same category are those that, when linked, release the most potential. Let us denote $ h(x) = d \cdot \log^2{(2+x)}$. 
Using the notation of function composition, let
$$h^{(0)}(x) = x , \;\;
h^{(i)}(x) = h\left( h^{(i-1)}(x) \right).$$
The category of a node is based on the values $h^{(i)}(\log{n}),$ $i=1,\dots,\log^{\ast}{n}$.
Note that $h^{(0)}(\log{n}) = \log{n},\: h^{(1)}(\log{n}) = d \cdot \log^2{(2+\log{n})},\dots,h^{(\log^{\ast}{n})}(\log{n}) = O(1)$, where the $O(1)$ depends on $d$, since (using the \emph{star} notation) $h^{\ast}(n) \leq \left( \log^{3} \right)^{\ast}(n) + O(1) = \log^{\ast}{n} + O(1)$.
\begin{definition}[Category]\label{catdef}
Let $u$ be a node. For $i=1,\dots,\log^{\ast}{n}$, we let  $\ct(u)=i$ if:
$$H(u.\lp) \in (h^{(i)}(\log{n}),h^{(i-1)}(\log{n})].$$
If $H(u.\lp) \leq h^{(\log^{\ast}{n})}(\log{n})$ we say that $u$ is of category $0$.
\end{definition}

The following crucial observations connect categories and savings in potential.

\begin{lemma}\label{equalcat}
Let {link}$(u,v)$ be type-(3). If $\ct(u) = \ct(v) \neq 0$, then the link decreases the potential by at least $1$.
\end{lemma}

\begin{proof}
Note that if $i=\ct(u) = \ct(v) \neq 0$ then
$$H(u.\lp) \geq h^{(i)}\left( \log{n} \right) \geq d \cdot \log^{2}{\left( 2 + H(v.\lp)  \right)},$$
$$H(v.\lp) \geq h^{(i)}\left( \log{n} \right)  \geq d \cdot \log^{2}{\left( 2 + H(u.\lp )  \right)}.$$
Thus, by Corollary~\ref{cor-categories}, the claim follows. 
\end{proof}
\begin{lemma}\label{cat0bound}
In each pairing round there are at most $O(\log{n})$ nodes of category $0$.
\end{lemma}

\begin{proof}
Let $x$ be of category $0$, then $H(x.\lp) = O(1)$. Denoting $a = s(x.\lp)$, $c = s(x.\rp$), we get
$H(x.\lp) = \log{\frac{a+c+1}{a}} = O(1)$. Therefore, $a = \Omega (c)$, an occurrence that can happen at most $O(\log{n})$ times in each round (by the same argument as in Lemma~\ref{lem5}). 
\end{proof}

\begin{lemma}\label{catchange}
Let $w$ denote the ``winner'' of linking $x$ and $y$ (neither of category $0$), i.e., $w$ is the one with the smaller key. Then $\ct(w)\geq \max\{\ct(x),\ct(y)\}$.
\end{lemma}

\begin{proof}
Let $y = x.\rp$, $a = s(x.\lp), b = s(y.\lp), c = s(y.\rp$) as in Figure~\ref{fig1}. We have that
$H(x.\lp) = \log{\frac{a+b+c+2}{a}}$, 
$H(y.\lp) = \log{\frac{b+c+1}{b}}$, and 
$H(\lnk(x,y).\lp) = \log{\frac{a+b+c+2}{a+b+1}}$.

Clearly $\frac{a+b+c+2}{a+b+1} \leq
\min\{ \frac{a+b+c+2}{a} , \frac{b+c+1}{b}  \}$,
finishing the proof. 
\end{proof}

As seen in Figure~\ref{fig_mp}, a delete-min operation transforms the ``spine'' of the heap (in binary view) into a balanced tree. We denote this tree by $T$. Each level of $T$ corresponds to a pairing round; specifically, level $i$ of $T$ consists of nodes at distance $i$ from the leaves, containing the \emph{losers} of the $i$-th pairing round.  
The following lemma captures the potential reduction that yields the main result.

\begin{lemma}\label{catmain}
Let $T'$\,be a subtree of $T$\,of depth $\log^{\ast}{n}$, whose leaves correspond to $2^{\log^{\ast}{n}}$ consecutive link operations. If $T'$\,contains only type-(3) links and no links involving nodes of category $0$, then the total decrease in potential caused by the links of $T'$\,is at least $1$.
\end{lemma}

\begin{proof}
Assume towards contradiction that there is no link between two nodes of the same category in $T'$. By Lemma~\ref{catchange} in each round the minimal overall category increases by at least 1, leaving us with two nodes of maximal category in the last round, a contradiction. By Lemma~\ref{equalcat}, a link between nodes of equal category decreases the potential by at least $1$. 
\end{proof}

\begin{theorem}\label{main-thm}
The amortized time of delete-min in multipass pairing heaps is $O(\log{n} \cdot \log^{\ast}{n}
\cdot 2^{\log^{\ast}{n}})$.
\end{theorem}
\begin{proof}
Let the real cost (number of link operations) be $k$. Note that there are at most $\lceil{\log{k}}\rceil$ pairing rounds.

Thus, if $k \leq \log{n} \cdot \log^{\ast}{n} \cdot 2^{\log^{\ast}{n}}$, then there are at most $\log{\log{n}} + \log{\log^{\ast}{n}} + \log^{\ast}{n} + 1$ rounds. Using Lemma~\ref{lemlevelpot} we get that the first $\log{\log^{\ast}{n}} + \log^{\ast}{n} + 1= O(\log^{\ast}{n})$ pairing rounds increase the potential by at most $O(\log{n} \cdot \log^{\ast}{n})$. Also, as shown in Lemma~\ref{lempot}, the total increase in potential for the last $\log{\log{n}}$ levels is $O(\log{n})$.
Thus, the total potential increase is at most $O(\log{n})$ + $O(\log{n} \cdot \log^{\ast}{n} )$.

To analyse the case $k > \log{n} \cdot \log^{\ast}{n} \cdot 2^{\log^{\ast}{n}}$, we use the potential decrease of type-(3) links. First, we look at the first $ \log^{\ast}{n} $ pairing rounds. 

By Lemma~\ref{catmain}, the links in every complete subtree of $T$ of depth $\log^{\ast}{n} $, in which there are only type-(3) links and no category-$0$ nodes, decrease the potential by at least $1$.

In the first $\log^{\ast}{n}$ levels of $T$ we can find $\frac{k}{ 2^{\log^{\ast}{n}}}$ \emph{disjoint} subtrees of this size. In these levels there are at most $O(\log^{\ast}{n} \cdot \log{n})$ type-(1),(2) links, or links containing category-$0$ nodes (Lemmas \ref{lem5} and \ref{cat0bound}). Thus, at least $\frac{k}{ 2^{\log^{\ast}{n}}} - O \left( \log^{\ast}{n} \cdot \log{n} \right)$ of the subtrees answer the conditions of Lemma~\ref{catmain}, decreasing the potential by at least $\frac{k}{ 2^{\log^{\ast}{n}}} - O \left( \log^{\ast}{n} \cdot \log{n} \right)$. Also, the total increase in potential caused by type-(1),(2) links is at most $O(\log{n} \cdot \log^{\ast}{n})$ (Lemma~\ref{lemlevelpot}).
Therefore, the first $\log^{\ast}{n}$ levels give us a decrease in potential of at least $\frac{k}{ 2^{\log^{\ast}{n}}} - O\left( \log^{\ast}{n} \cdot \log{n} \right)$.

Note that by using the same argument on the next $\log^{\ast}{n}$ levels, we get a decrease in potential of at least $\frac{k'}{ 2^{\log^{\ast}{n}}} - O\left( \log^{\ast}{n} \cdot \log{n} \right)$, where $k'$ is the number of links in level $\log^{\ast}{n} + 1$. 
Thus, levels which contain $\Omega \left( \log{n} \cdot \log^{\ast}{n} \cdot 2^{\log^{\ast}{n}} \right)$ links only decrease the potential. 

We repeat this argument until we reach a level in $T$ containing $\tilde{k} \leq \log{n} \cdot \log^{\ast}{n} \cdot 2^{\log^{\ast}{n}}$ links. Now, applying the same argument as for the first case, we get that the total increase in potential for the last $\log{\tilde{k}}$ levels (starting from the level of $\tilde{k}$ links) is at most $O(\log{n} \cdot \log^{\ast}{n})$.

Summarizing, the total amortized time (in both cases) is at most $$k + O(\log{n} \cdot \log^{\ast}{n}) - \left( \frac{k}{ 2^{\log^{\ast}{n}}} - \log^{\ast}{n} \cdot \log{n} \right).$$ Scaling the potential by $2^{\log^{\ast}{n}}$, we get that the amortized time is $O(\log{n} \cdot \log^{\ast}{n} \cdot 2^{\log^{\ast}{n}})$. 
\end{proof}

\section{Path-balanced binary search trees}
\label{sec:pb}

Consider the operation of accessing a node $x$ in a BST $T$ with $n$ nodes (we refer interchangeably to a node and its key). Let $\P^{x}$ denote the search path to $x$ (i.e., the path from the root of $T$ to $x$). The path-balance method re-arranges $\P^x$ into a complete balanced BST (with all levels complete, except possibly the lowest). Subtrees hanging off $\P^x$ are re-attached in the unique way given by the key-order (Figure~\ref{fig_pb}). There are multiple ways to implement this transformation such that the number of pointer moves and pointer changes is linear in the length of the search path. For instance, we may first rotate the search path into a \emph{monotone} path, then apply a  \emph{multipass transformation} (described next) to this monotone path. 


\subparagraph{Multipass transformation.}
A multipass transformation of a monotone path $\P$ (of which the deepest node might not be a leaf) converts $\P$ into a balanced tree (in which the last level may be incomplete) by a sequence of \emph{pairing rounds}. In each pairing round we rotate every other edge in a prefix of $\P$ (i.e., a subpath of the shallowest nodes on $\P$). Each rotation pushes one node off $\P$. We denote by $\P^i$ the path remaining of $\P$ after $i$ pairing rounds. The pairing rounds are defined as follows. We assume that the path consists of right child pointers; in the case it consists of left child pointers everything is symmetric.

Let $\ell (\P)$ denote the length of $\P$ (i.e., the number of nodes on $\P$). In the first round we do just enough rotations so that the length of the path after the round (i.e., $\P^{1}$) is one less than a power of $2$. Specifically, we do $\alpha$ rotations where $\alpha$ is the smallest integer such that $\ell(\P^1) = \ell(\P)-\alpha = 2^j-1$. In the second round we do $2^{j-1}-1$ rotations on $\P^1$, and in round $i>1$ we do $2^{j-i+1}-1$ rotations on $\P^{i-1}$. We maintain the invariant that after $i+1$ rounds all the nodes that were pushed off $\P$ (excluding those that were pushed off $\P$ at the first round) are arranged in balanced binary trees of height $(i-1)$, hanging as children of the nodes of $\P^{i+1}$.  

The proof of the following theorem is analogous to the proof of Theorem~\ref{main-thm} (one can verify that all steps of the proof still hold for the slightly modified pairing rounds of the multipass transformation, replacing rotations by links).

\begin{theorem}\label{multipass-transformation}
For every monotone path $\P$ with $\ell(\P)=k$, the change in $\Phi$ caused by applying a multipass transformation on $\P$ is bounded by
$\displaystyle \Delta \Phi \leq c(n,k) := -\frac{k}{2^{\log^{\ast}{n}}} + O(\log{n} \cdot \log^{\ast}{n}),\label{MT-cost}
$
where $n$ is the size of the subtree of the root of $\P$.
\end{theorem}

\subparagraph{Warm-up: a simplified path-balance.}
We first look at an easier-to-analyse variant of path-balance, where, instead of a complete balanced tree, we build an \emph{almost} balanced tree out of the search path $\P^x$, as follows: we first make the accessed item $x$ the root, then turn the parts of $\P^x$ containing items smaller (resp.\ larger) than $x$ into balanced subtrees rooted at the left (resp.\ right) child of $x$. The depth of this tree is at most one larger than the depth of a complete balanced tree built from $\P^x$. 

For the purpose of the analysis, we view the simplified path-balance transformation as a two-step process (Figure~\ref{naive-path-balance}). The actual implementation may be different but the analysis applies as long as the transformation takes time $O\left(\ell(\P^x)\right)$.

{\bf Step 1.} Rotate the accessed element $x$ all the way to the root. (Observe that after this step, $\P^x$ is split into two monotone paths, $\P^{<x}$ to the left of $x$ consisting only of ``right child'' pointers, and $\P^{>x}$ to the right of $x$, consisting only of ``left child'' pointers.)
 
{\bf Step 2.} Apply a multipass transformation to $\P^{>x}$ and to $\P^{<x}$.  

We show that the amortized time of an access using simplified path-balance is $O(\log{n} \cdot \log^{\ast}{n}
\cdot 2^{\log^{\ast}{n}})$. We use the same potential function as in \S\,\ref{sec:mp}, and we assume the two-step implementation described above. We first state an easy observation.

\begin{lemma}\label{path-pot}
Let $\P$ be a path in $T$ rooted at a node $r$, then $\Phi \left( \P \right) = O(\log{s(r)})$, where $\Phi(\P) = \sum_{x\in \P}{\phi(x)}$ and $s(r)$ is the size of the subtree of $r$.
\end{lemma}

\begin{proof}
Denote $\ell = \ell(\P)$. Let $a_1\leq ... \leq a_{\ell}=s(r)$ be the subtree-sizes of the nodes on $\P$ from the deepest node to $r$. Then
\begin{align*}
\Phi(\P) = \sum^{\ell-1}_{k=1} f \left( \frac{a_{k+1}}{a_k} \right) =
\sum^{\ell-1}_{k=1} g \left( \log{ \frac{a_{k+1}}{a_k} }\right) \leq \ell \cdot g \left(\frac{\log{s(r)}}{\ell} \right) =  O(\log{s(r)}),
\end{align*}
due to $g$'s concavity and since the terms $\log{ \frac{a_{k+1}}{a_k} }$ sum to $\log{s(r)}-\log{a_1} \leq \log{s(r)}$.
\end{proof}

We proceed with the analysis. We argue that rotating $x$ to the root (Step 1) increases $\Phi$ by at most $O(\log{n})$.
To see this, observe first, that the potential of nodes hanged on the nodes of $\P^{x}$ excluding $x$, can only decrease. This is because their subtree remains the same, whereas the subtree of their parent (a node on the search path) can only lose elements, (see Figure~\ref{naive-path-balance}). The two children of $x$ may increase the potential by at most $O(\log{n})$.

For nodes \emph{on the search path}, we look at the potential after the transformation. We have two separate paths (see Figure~\ref{fig2} middle), and by Lemma~\ref{path-pot} the potential of each path is bounded by $O(\log{n})$. This concludes the analysis for Step 1.

In Step~2, as we apply the multipass transformation to both $\P^{<x}$ and $\P^{>x}$, Theorem~\ref{multipass-transformation} applies. Thus, $\Delta \Phi$ is at most $c( s(x.\lp), \ell(\P^{<x})) + c(s(x.\rp),\ell(\P^{>x}))$ where $c(n,k)$ is defined in Theorem~\ref{MT-cost}. The claim on the amortized running time follows by scaling $\Delta\Phi$ by $2^{\log^{\ast}{n}}$ and adding it to the actual cost (the length of $\P^x$).
This concludes the proof.

\subparagraph{Analysis of path-balance.} The original path-balance heuristic (where we insist on building a \emph{complete} balanced tree) is trickier to analyse. Here, instead of moving the accessed item $x$ to the root, we move the \emph{median} item $m$ of the search path $\P^x$ to the root. Here, ``median'' is meant with respect to the ordering of keys; $m$ is, in general, \emph{not} the node with median depth on $\P^x$. It is instructive to prove the earlier $O(\log{n} \cdot \log\log{n} / \log\log\log{n})$ result first, by re-using parts of the Fredman et al.\ proof for multipass. We do this in Appendix~\ref{section-raman}. In the remainder of this section we prove the new, stronger result.

\begin{theorem}\label{thmBST}
The amortized time of search using path-balance is
$O \left( \log{n} \cdot \left( \log^{\ast}{n} \right)^2 \cdot 2^{\log^{\ast}{n}} \right)$.
\end{theorem}

For the purpose of the analysis, we view the path-balance transformation as a sequence of recursive calls on search paths in \emph{some} subtree of $T$.  
The total \emph{real} cost is proportional to the original length of the search path to $x$ which we denote by $k$.
We define a threshold $\tau = \log{n}$, and distinguish between recursive calls on paths shorter than $\tau$ (``short paths'') and recursive calls on paths longer than $\tau$ (``long paths''). 

 A \textbf{long path} $\P^x$ is processed as follows. We rotate the median $m$ of the nodes on $\P^x$ to the root, splitting $\P^x$ into two paths of equal lengths. One of these paths contains the path from $m$ to $x$ in $\P^x$, and the other path, which is monotone, contains either the elements smaller than $m$ on $\P^x$ or the elements larger than $m$ on $\P^x$ (depending upon whether $x$ is in the right or left subtree of $m$). In the sequel we assume without loss of generality that the monotone part contains all elements larger than $m$ and denote it by $\P^{>m}$. We denote the other (non-monotone) path that ends with $x$ by $Q^x$. We perform a multipass transformation on $\P^{>m}$, and make a recursive call on $Q^x$ (i.e., $Q^x$ becomes the $P^x$ of the next recursive call); see Figure~\ref{fig2}.

A \textbf{short path} $\P^x$ is transformed into a balanced binary tree in two phases, as follows.  
In the first phase, we rotate up the median $m_1$ of $\P^x=\P^1$ until it becomes the root of the subtree rooted at the shallowest node of $\P^1$. This decomposes $\P^1$ into a monotone path and a general path $\P^2$, one starting at the left child of $m_1$ and the other at the right child of $m_1$. We repeat this recursively with the median $m_2$ of $\P^2$, and so on, until we get a general path $\P^{\ell}$ of length $1$. After this transformation, the medians $m_j$ form a path, each $m_j$ having the next median $m_{j+1}$ as one child and a monotone path as the other child. The lengths of these monotone paths decrease exponentially by a factor of $2$. In the second phase we apply a multipass transformation on each of the monotone paths, obtaining a complete balanced tree; see Figure~\ref{figlast}.

Before we analyse each case, we argue that Theorem~\ref{multipass-transformation} also holds with a modified potential $\Phi$ (defined below). As we only use the new potential from now on, there is no risk of confusion. The modification consists in changing the exponent of the logarithmic term in the denominator from $2$ to $3$, and changing the additive constant inside the $\log(\cdot)$ to make sure $\Phi$ is still increasing everywhere.

Formally, $\Phi = \sum_{x\in T}{\phi(x)}$, where $\phi(x) = \frac{H(x)}{\log^3{(4+H(x))}}$, and $H(x) = \log{\frac{s(p(x))}{s(x)}}$. As earlier, $s(x)$ is the size of the subtree rooted at $x$, and $p(x)$ is the parent of $x$. For convenience, we define the functions $f(x) = \log{x}/\log^3{(4+\log{x})}$, and $g(x) = x/\log^3{(4+x)}$. As before, $f(x) = g(\log{(x)})$, and $\phi(x) = f(\frac{s(p(x))}{s(x)})$.

We show in Appendix~\ref{second-Appendix} that the entire analysis in \S\,\ref{sec:mp} extends to this new potential. Therefore, Theorem~\ref{multipass-transformation} holds also for the modified potential function $\Phi$. 

Now, the analysis of transforming long paths is straightforward. For short paths, we need two new observations. \\

\begin{lemma}\label{lempot2}
The total increase in potential for performing multipass transformation on a path $\P$ of length $k < \log{n}$ where $n$ is the size of the subtree of the root of $\P$, is at most
\begin{align*}
\sum_{j=1}^{\log{k}}\frac{O(\log{n})}{(\log\log{n}+1-j)^3}.
\end{align*}
\end{lemma}
The proof is identical to that of Lemma~\ref{lempot}. As before, the sum can be bounded as $O(\log{n})$, but here we use the quantity explicitly inside another sum where the exponent $3$ in the denominator will be crucial. The next observation can be shown in a way similar to Lemma~\ref{path-pot}.

\begin{lemma} [Appendix~\ref{last_proof}]\label{multi_median}
Given a search path $\P$ of length $k < \log{n}$, the total increase in $\Phi$ due to recursively rotating all medians $m_1, m_2, \dots$ of $\P$ to the root is $O(\log{n})$. 
\end{lemma}

We are ready to prove Theorem~\ref{thmBST}. We split the proof into three cases according to the length of the search path, denoted by $k$.

\textbf{Short paths ($k \leq \tau = \log{n}$).} Notice that $\log{k} \leq \log{\log{n}}$.
Recall that in the first phase, we repeatedly rotate up the medians, decomposing the path into monotone paths of lengths $1,2,4,\dots,2^{j}$, where $j < \log{\log{n}}$. By Lemma~\ref{multi_median} the total increase in potential due to this transformation is at most $O(\log{n})$.

In the second phase, we do a multipass transformation on each of these monotone paths. By Lemma~\ref{lempot2}, a multipass transformation on a monotone path of length $2^j$ increases $\Phi$ by at most $\sum_{i=1}^{j}{\alpha \cdot\log{n} / {(\log{\log{n}} + 1 - i )^3}}$, for some fixed $\alpha$. Thus, the $j <\log{\log{n}}$ multipass transformations increase the potential by at most

$$ \sum_{j=1}^{\log{\log{n}}}{\sum_{i=1}^{j}{ \frac{\alpha \cdot\log{n} }{(\log{\log{n}} + 1 - i )^3}}} =
\sum_{s = 1}^{\log{\log{n}}}{\frac{\alpha \cdot\log{n} }{s^2}} = O(\log{n}).$$

The first equality holds since the term $\frac{\alpha \cdot\log{n} }{s^3}$ appears in the above sum exactly $s$ times ($1 \leq s \leq \log{\log{n}}$). Thus, the total increase in $\Phi$ is, in this case, $O(\log{n})$.

\textbf{Longish paths ($\tau < k \leq  \log{n} \cdot \log^{\ast}{n} \cdot 2^{\log^{\ast}{n}}$).} Notice that $\log{k} \leq \log{\log{n}} + 2\cdot\log^{\ast}{n}$.

We perform $2\cdot\log^{\ast}{n}$ recursive calls and a final call on a search path of length $k' \leq \tau$. The final call increases $\Phi$ by at most $O(\log{n})$, by the analysis in the previous case. The recursive calls consist of rotating the current median up to the root and applying the multipass transformation on a monotone path. As before, rotating the median up increases $\Phi$ by at most $O(\log{n}$). Also, each multipass transformation is performed on a path of length $\leq \log{n} \cdot \log^{\ast}{n} \cdot 2^{\log^{\ast}{n}}$. By Theorem~\ref{multipass-transformation}, the increase in potential is at most $O(\log{n} \cdot \log^{\ast}{n})$.
Therefore, the $2\cdot\log^{\ast}{n}$ recursive calls increase $\Phi$ by at most  $O\left(\log{n} \cdot \left( \log^{\ast}{n} \right)^2 \right)$, which also bounds the total increase in $\Phi$.

\textbf{Long paths} ($k = \Omega \left( \log{n} \cdot \log^{\ast}{n} \cdot 2^{\log^{\ast}{n}} \right)$). We look at the potential change due to the first recursive call. Again, rotating the median $m$ to the root increases $\Phi$ by at most $O(\log{n})$. The path splits into $\P^{>m}$ and $Q^{x}$, of which $\P^{>m}$ is monotone. By Theorem~\ref{multipass-transformation}, the multipass transformation on $\P^{>m}$ \emph{decreases} $\Phi$ by $\frac{k/2}{ 2^{\log^{\ast}{(n)}}} - O \left( \log^{\ast}{(n)} \cdot \log{n} \right)$. 

By the same argument, $\Phi$ decreases during all of the subsequent recursive calls on paths of size $\Omega \left(\log{n} \cdot \log^{\ast}{n} \cdot 2^{\log^{\ast}{n}} \right)$.

We continue until we have a recursive call on a path of size at most $\left( \log{n} \cdot \log^{\ast}{n} \cdot 2^{\log^{\ast}{n}}\right)$, which, by the previous case, increases $\Phi$ by at most $O\left( \log{n} \cdot \left( \log^{\ast}{n} \right)^2 \right)$. Thus, we obtain that the total decrease in $\Phi$ in this case is at least $\frac{k/2}{ 2^{\log^{\ast}{(n)}}} - O\left( \log{n} \cdot \left( \log^{\ast}{n} \right)^2 \right)$.

Combining the three cases, after scaling the potential by $2 \cdot 2^{\log^{\ast}{n}}$, we conclude that the amortized time of the access is $k + 2 \cdot 2^{\log^{\ast}{n}} \cdot \Delta \Phi = O\left( \log{n} \cdot 2^{\log^{\ast}{n}} \cdot \left( \log^{\ast}{n} \right)^2 \right)$, as required.

\newpage
\appendix
\section{Additional figures}

\begin{figure}[H]
    \begin{center}
\includegraphics[width=0.75\textwidth]{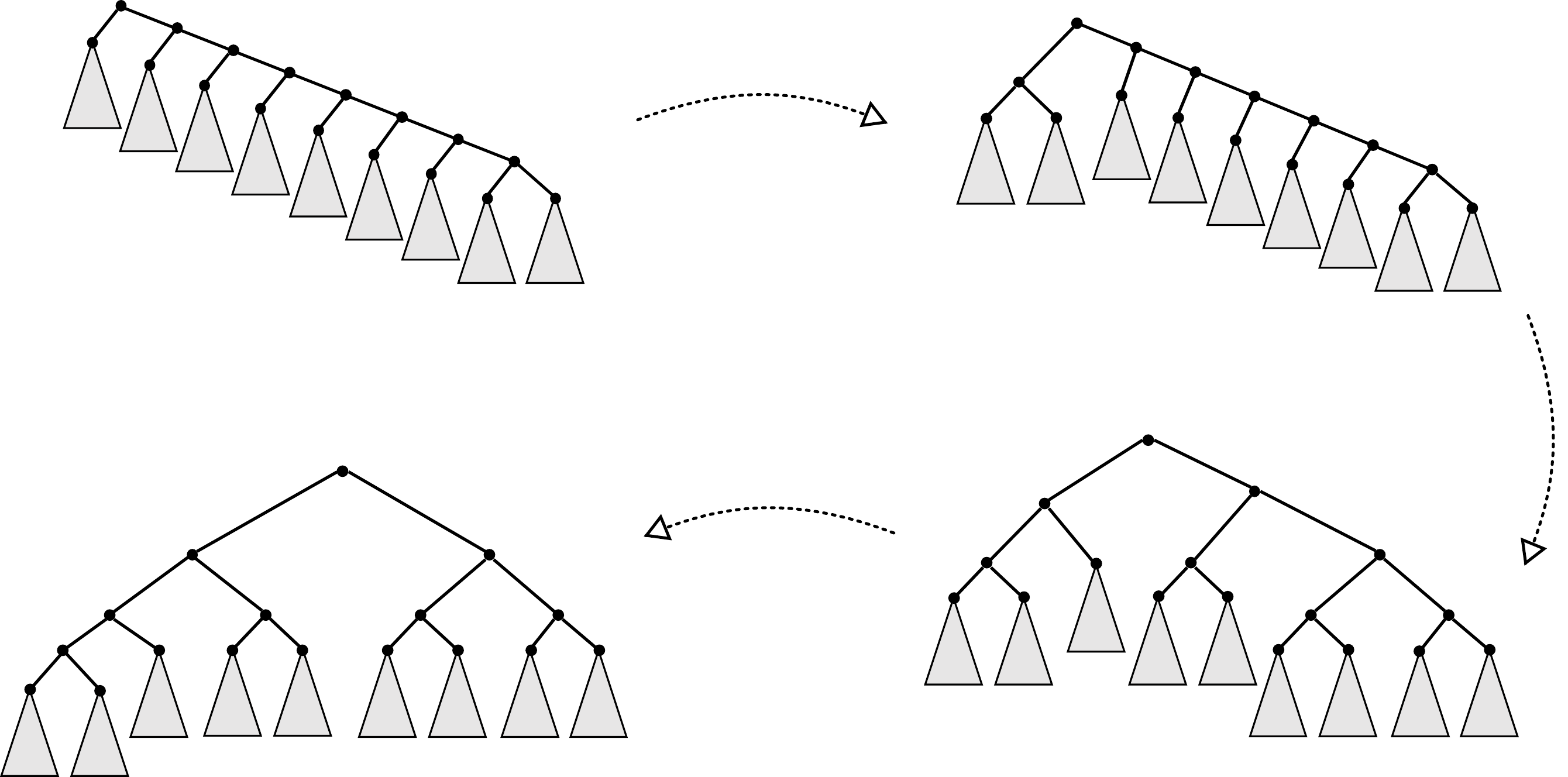}
    \end{center}
\caption{Example multipass transformation. In the first round we make the length of the path $2^3-1$. The next steps are similar to \emph{delete-min} in multipass pairing heaps, except that in order to build a balanced tree, we do not link the last nodes. 
\label{fig_MT}}
\end{figure}

\begin{figure}[H]
    \begin{center}
\includegraphics[width=0.95\textwidth]{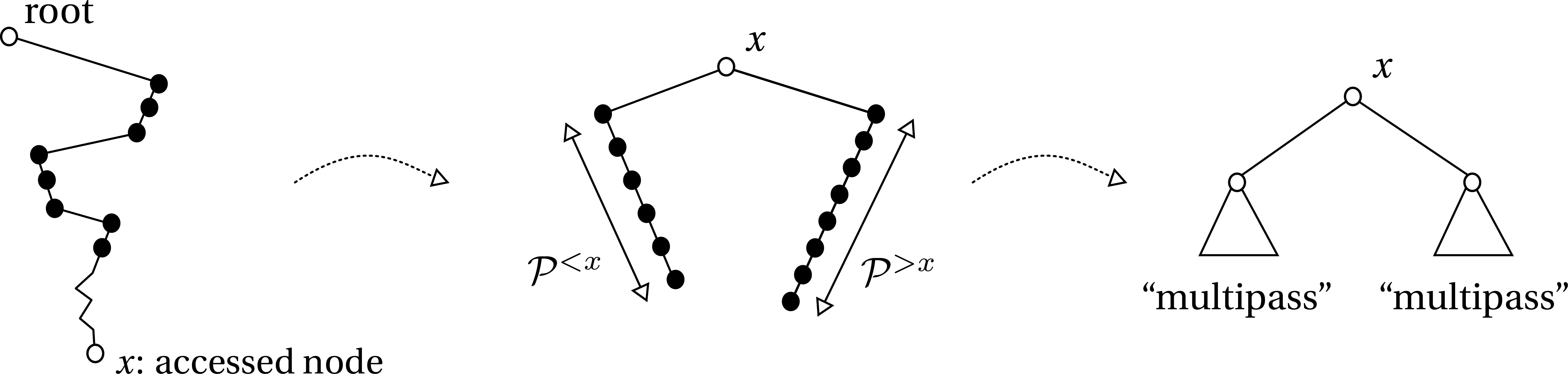}
    \end{center}
\caption{Two-step view of simplified path-balance restructuring.
\label{naive-path-balance}}
\end{figure}

\begin{figure}[H]
    \begin{center}
\includegraphics[width=1\textwidth]{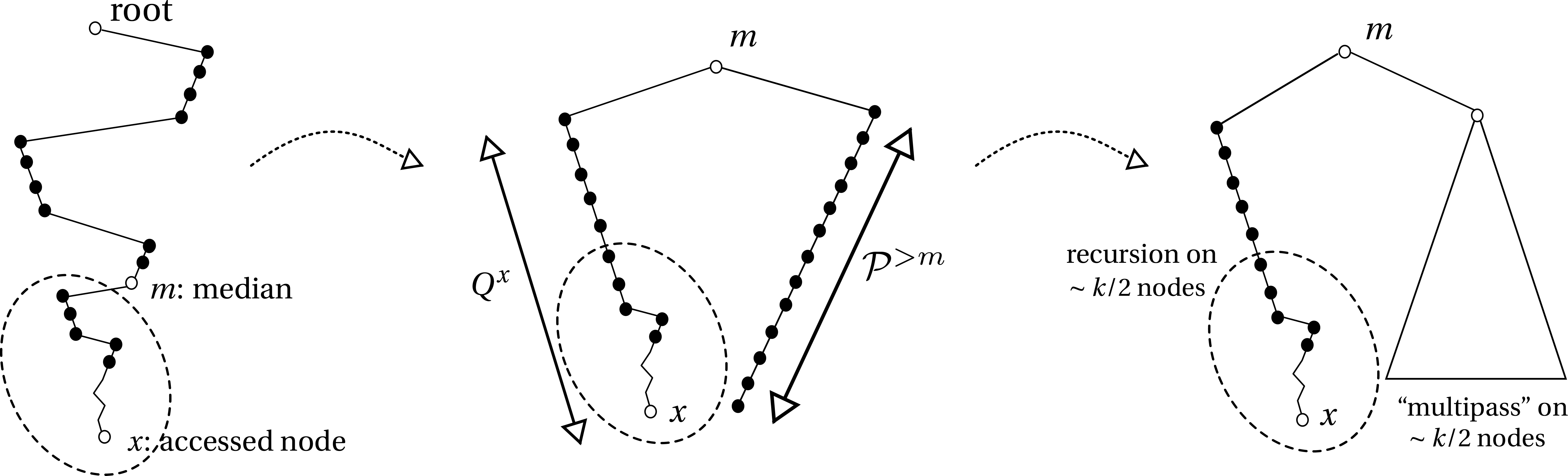}
    \end{center}
\caption{Recursive view of path-balance restructuring.
\label{fig2}}
\end{figure}

\begin{figure}
    \begin{center}
\includegraphics[width=0.9\textwidth]{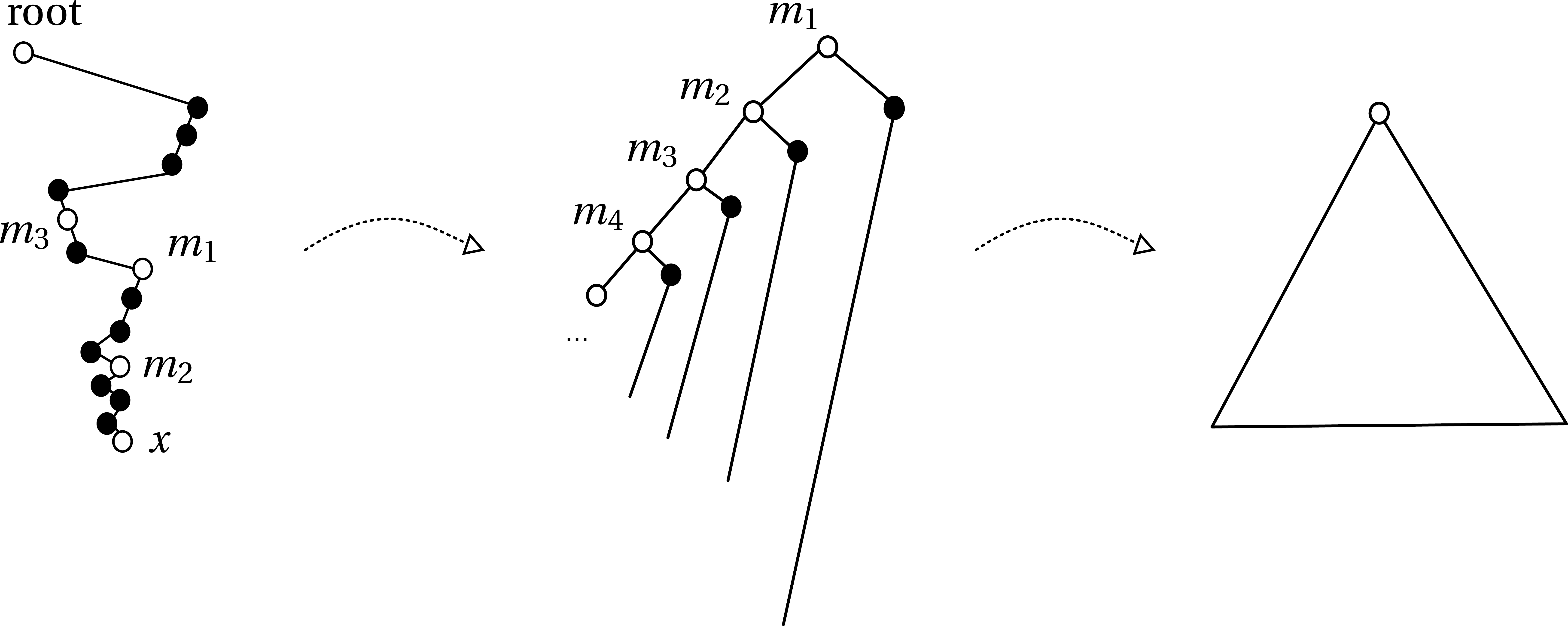}
    \end{center}
\caption{Two-step view of path-balance on a path shorter than $\tau$. First, we recursively rotate up medians. The monotone paths do not necessarily follow left child pointers as shown in this figure but can also follow right child pointers. Finally, we apply a multipass transformation to each monotone path, obtaining a complete balanced tree.
\label{figlast}}
\end{figure}

\newpage
\section{Additional proofs for \S\,\ref{sec:mp}}
\label{appa}
We start with a multi-part technical lemma, to be used in the proofs of other claims.

\begin{lemma}\label{lem123} 
\begin{enumerate}[(i)] \ \\
\item For every $x \geq 1,y \geq 0$, it holds that $f(x+y)\leq f(x)+1.5y.$
\item For every $x,y \geq 1$, it holds that $f(xy)-f(x)\leq f(y).$
\item For every $x\geq \gamma$ it holds that
$$f'(x) \geq \frac{1}{3x\log^2(2+\log{x})}.$$
\item For every $x>y\geq 1, c > 0$, it holds that $f(x+c)-f(y+c)\leq f(x)-f(y).$

\item $f(1+x) + f(1+\frac{1}{x})$ has only one maximum point at $x=1$ in
$[\frac{1}{\gamma},\gamma]$, and two global minima.
\item Fix $a,b \in \mathbb{N}$. If $\frac{1}{\gamma} \leq \frac{a}{b} \leq \gamma$, then:
$$f \left( \frac{a+b+1}{a} \right) +
f \left( \frac{a+b+1}{b} \right) \leq 0.95.$$
\end{enumerate}
\end{lemma}

\begin{proof}
\ \\
\noindent Part (i):

By Lagrange theorem, $f(x+y) = f(x) + y\cdot f'(c)$, for $c\in \left[x,x+y \right]$.
Due to the concavity of $f$, we get $f'(c) \leq f'(1) = \frac{1}{\ln{2}} < 1.5$.

\noindent Part (ii):

It is enough to show $g(\log{x} +\log{y}) - g(\log{x}) \leq g(\log{y})$, which holds since $$g(x+y) = g(x) + \int_{0}^{y}{g'(x+t)} \leq g(x) + \int_{0}^{y}{g'(t)} = g(x) + g(y), $$
where the inequality holds since $g$ is concave.

\noindent Part (iii):

Taking the derivative:
\begin{align*}
f'(x) &=
\frac{1}{x \cdot \ln{2} \cdot \log^2(2+\log{x})}
-\frac{\frac{2 \log{x}}{\log{x}+2}}
{x \cdot (\ln{2})^2 \cdot \log^3(2+\log{x})} \\
&\geq
\frac{1}{x \cdot \ln{2} \cdot \log^2(2+\log{x})}
-\frac{2}{x \cdot (\ln{2})^2 \cdot \log^3(2+\log{x})} \\
&=
\frac{1}{x\log^2(2+\log{x})} \left[ \frac{1}{\ln{2}} - \frac{2}{(\ln{2})^2 \cdot \log(2+\log{x})} \right] \\
&\geq
\frac{1}{x\log^2(2+\log{x})} \left[ \frac{1}{\ln{2}} - \frac{2}{(\ln{2})^2 \cdot \log(2+\log{\gamma})} \right] \geq
\frac{0.335}{x\log^2(2+\log{x})}.
\end{align*}

\noindent Part (iv):

It is equivalent to prove $f(x+c)-f(x) \leq f(y+c)-f(y)$, i.e., that $g(z) = f(z+c) - f(z)$ is monotone decreasing. This holds since $g'(z) = f'(z+c) - f'(z) < 0 $ ($f$ is concave).

\noindent Part (v):

Because of symmetry around $x=1$, it suffices to prove that $h(x) = f(1+x) + f(1+\frac{1}{x})$ has only one minimum in $[1,\gamma]$. This minimum is at $x \approx 24.271$. A plot of this function is shown in Figure~\ref{wolfram}. We omit the tedious analytical derivation. 

\begin{figure}
\centering
\begin{subfigure}{0.48\linewidth}
\includegraphics[width=0.95\textwidth]{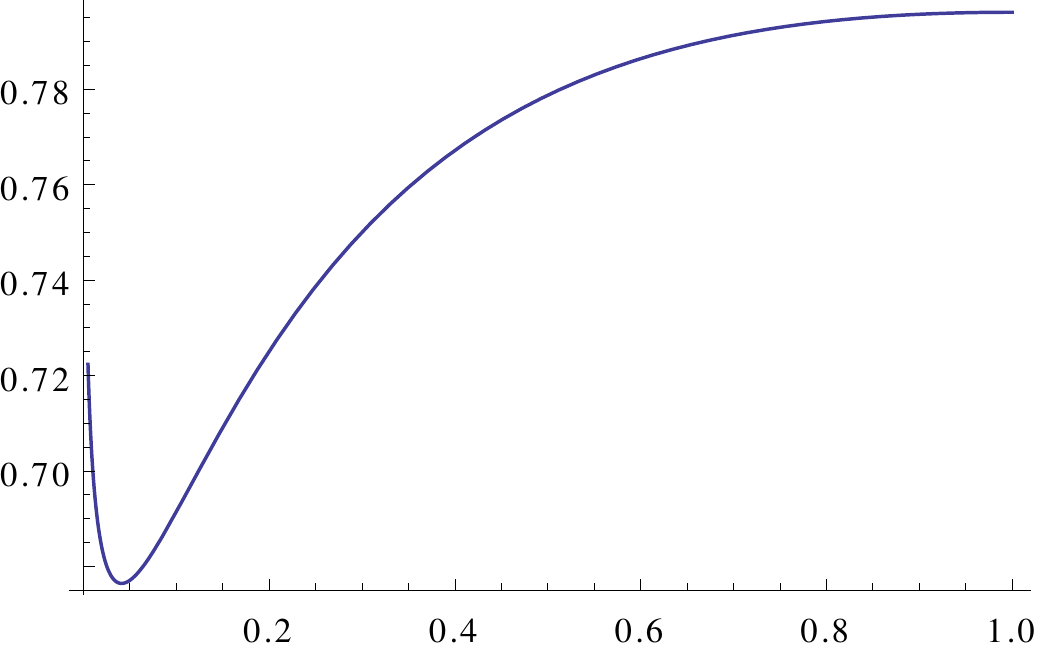}
\end{subfigure}~
\begin{subfigure}{0.48\linewidth}
\includegraphics[width=0.95\textwidth]{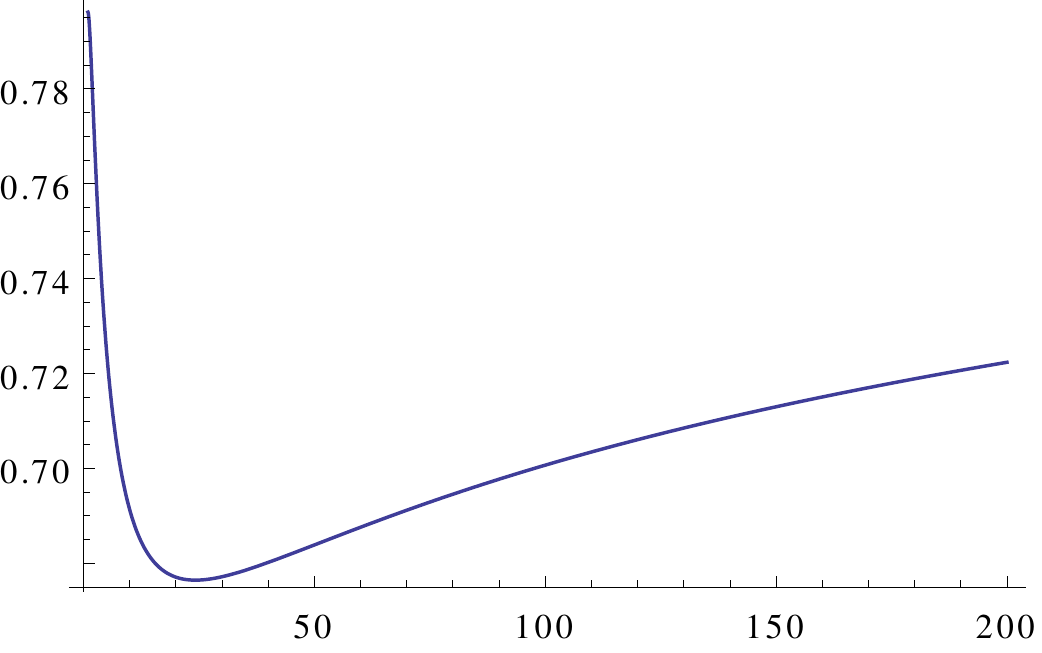}
\end{subfigure}

\caption{Plot of $h(x)=f(1+x) + f(1+\frac{1}{x})$ for $x\in [\frac{1}{200},1]$ (left), and $x\in [1,{200}]$ (right).
\label{wolfram}}
\end{figure}

\noindent Part (vi):

For $1 \leq a,b \leq 3000$ we verified the inequality by computer (the maximal value is $\approx 0.934$, at $a=b=1$). Thus, assume $a,b > 3000$. Denote $x = \frac{a}{b}$. Using Lemma~\ref{lem123}(i), we get:
\begin{align*}
&f \left( \frac{a+b+1}{a} \right) +
f \left( \frac{a+b+1}{b} \right) =
f \left( 1 + \frac{1}{x} + \frac{1}{a} \right) +
f \left( 1 + x + \frac{1}{b} \right) \\
&\leq f \left( 1 + \frac{1}{x} \right) + f \left( 1 + x \right) + 1.5 \left( \frac{1}{a} + \frac{1}{b} \right) \leq
f \left( 1 + \frac{1}{x} \right) + f \left( 1 + x \right) + 0.001.
\end{align*}
Using Lemma~\ref{lem123}(v) we further bound the expression above:
$$f \left( 1 + \frac{1}{x} \right) + f \left( 1 + x \right)
\leq \max \{f(2) + f(2), f\left(1+\frac{1}{\gamma}\right) + f(1+\gamma) \} \leq 0.82.$$
Thus, we are done. 
\end{proof}

\subsection{Proof of Lemma~\ref{lema}}\label{lem1proof}

\begin{proof}
\textbf{Case $\gamma^2 a \geq \gamma b \geq c$:} By inspecting $\Delta \Phi$ in Equation (\ref{eq:1}) we observe that we can drop the (positive) $1$-st and the (negative) $5$-th terms (as a negative quantity) and also the $3$-rd and $7$-th terms (since their sum is negative).
Denote $x = a/b \geq 1/\gamma$, and $y = b/c \geq 1/\gamma$. We have:

\begin{align*}
\Delta \Phi &\leq f\left(\frac{a+b+1}{b}\right)
+ f\left(\frac{a+b+c+2}{c}\right)
- f\left(\frac{a+b+c+2}{b+c+1}\right) -
f\left(\frac{b+c+1}{c}\right) \\
&\leq f\left(\frac{a+b+1}{b}\right) + f\left(\frac{a+b+c+2}{c}\right) \leq
f\left(\frac{a+b+b}{b}\right) + f\left(\frac{a+b+c+2c}{c}\right)
\\ &= f(2+x) + f(3+y+xy).
\end{align*}

Observe that $x,y \leq \gamma a/c$, and $xy = a/c$, and thus:
\begin{align*}
\Delta \Phi & \leq f\left( 2+ \gamma a/c\right) + f\left( 3 + \gamma a/c + a/c\right)\\
& \leq 2 f \left( 3 + (\gamma + 1) (a/c) \right)\\
& \leq 2 f \left( (3\gamma^2 + \gamma + 1) (a/c) \right)\\
& \leq 2 f \left( 4\gamma^4 (a/c) \right)\\
& \leq 2 g \left( \log\left(4\gamma^4 \right) + \log{(a/c)} \right),
\end{align*}

which is in the required form. \\

\textbf{Case $ \gamma^2 a \geq c \geq \gamma b$:} By inspecting $\Delta \Phi$ in Equation (\ref{eq:1}) and dropping the 1-st and 5-th terms (since their sum is negative), we obtain:
\begin{align}
\Delta \Phi \: \leq  \label{eq:2}
\begin{split}
&f\left(\frac{a+b+1}{b}\right) + f\left(\frac{a+b+c+2}{a+b+1}\right) + f\left(\frac{a+b+c+2}{c}\right)\\
- &f\left(\frac{a+b+c+2}{b+c+1}\right) - f\left(\frac{b+c+1}{b}\right) - f\left(\frac{b+c+1}{c}\right).
\end{split}
\end{align}

Denote $x = a/c \geq 1/\gamma^2$, and $y = c/b \geq \gamma$.

Suppose $x \leq 1$ (which is equivalent to $c \geq a$). We can drop the $1$-st and $5$-th terms in Equation (\ref{eq:2}) (since their sum is negative). Further dropping the negative terms, we are left with:
\begin{align*}
\Delta \Phi & \leq f\left(\frac{a+b+c+2}{a+b+1}\right) +  f\left(\frac{a+b+c+2}{c}\right)\\
& \leq f\left(2+\frac{c}{a+b}\right) +  f\left(\frac{a+b+c+2c}{c}\right)\\
& = f\left( 2 + y/(1+xy) \right) +  f\left( 3+ x + 1/y\right) \quad \tag{as $c=yb$ and $a=xyb$}\\
& \leq  f\left(2 + y/(1+y/\gamma^2) \right) + f(5) \quad ~\tag{as $1 \geq x \geq 1/\gamma^2$} \\
& \leq f\left(2 + \gamma^2 y/(1+y) \right) + f(5) \\
& \leq  f(2\gamma^2) + f(5) \leq 2\gamma^2 + 5.
\end{align*}

This completes the proof for $x \leq 1$.
Now assume that $x > 1$ (which is equivalent to $a>c$). Then we can drop the $2$-nd and $4$-th terms in Equation (\ref{eq:2}) (since their sum is negative). After dropping the $6$-th term, we are left with:
\begin{align*}
\Delta \Phi &\leq f\left(\frac{a+b+1}{b}\right) + f\left(\frac{a+b+c+2}{c}\right) - f\left(\frac{b+c+1}{b}\right)\\
&= f\left(\frac{a+b}{b} + \frac{1}{b}\right) -
f\left(\frac{b+c}{b} + \frac{1}{b}\right) +
f\left(\frac{a+b+c+2}{c}\right)\\
&\leq f\left(\frac{a+b}{b} \right) -
f\left(\frac{b+c}{b} \right) +
f\left(\frac{a+b+c+2}{c}\right) \quad \tag{by Lemma~\ref{lem123}(iv)}\\
&= f\left( 1 + xy\right)  - f\left( 1+ y\right) +
f\left( 1 + x + 1/y +2/c\right)\\
 & \leq f\left( 1 + xy\right)  - f\left( 1+ y\right) + f\left( 4+x\right).
\end{align*}

By Lemma~\ref{lem123}(ii) we bound the first two terms as: 
$$f\left( 1 + xy\right) - f\left( 1 + y\right)  \leq f\left( x + xy\right) - f\left( 1+ y\right)\\
 = f\left( x(1+y) \right)  - f\left( 1+ y\right) \leq f(x).$$

Thus, we have:
\begin{align*}
\Delta \Phi & \leq f(x) + f\left( 4 + x \right)\\
& \leq 2f(4+x)\\
& \leq 2f(4x+x)\\
& \leq 2f(5x)\\
& \leq 2g\bigl(\log{(5)}+\log{(a/c)}\bigr),
\end{align*}
in the required form. 

\end{proof}

\subsection{Proof of Lemma~\ref{lemb}}\label{lem2proof}

\begin{proof}
Using $b= \Omega(a+c)$ and collecting all constant terms in Equation (\ref{eq:1}) we get:
$$\Delta \Phi = f\left( \frac{a+b}{a} \right) + f\left( \frac{a+b+c}{c} \right) -
f\left( \frac{a+b+c}{a} \right) - f\left( \frac{b+c}{c} \right) + O(1)=$$
$$f\left( \frac{b}{a} \right) + f\left( \frac{b}{c} \right) - f\left( \frac{b}{a} \right) - f\left( \frac{b}{c} \right) + O(1) = O(1).$$
The second equality in the equation above holds since $f\left( O \left( x \right) \right) = f(x) +O(1)$.
\end{proof}

\subsection{Proof of Lemma~\ref{lemc}}\label{lemcproof}

\begin{proof}

\textbf{Type-(3A) link }($c \geq \gamma^2 a \geq \gamma b$):
Recall the change in potential from Equation (\ref{eq:1}).
\begin{align}
\Delta \Phi = \tag{\ref{eq:1}}
\begin{split}
f\left(\frac{a+b+1}{a}\right) + f\left(\frac{a+b+1}{b}\right) + f\left(\frac{a+b+c+2}{a+b+1}\right) + f\left(\frac{a+b+c+2}{c}\right) \\
- f\left(\frac{a+b+c+2}{a}\right) - f\left(\frac{a+b+c+2}{b+c+1}\right) - f\left(\frac{b+c+1}{b}\right) - f\left(\frac{b+c+1}{c}\right).
\end{split}
\end{align}
Denote $x=\frac{c}{a}\geq \gamma^2, y=\frac{a}{b}\geq \frac{1}{\gamma}$.
 We consider three cases according to the value of $y$.

{\bf Case $y\leq \gamma$:} 
We use Lemma~\ref{lem123}(vi) to bound the first two terms by $0.95$. The $7$-th term is larger than the $3$-rd so if we discard them both we only increase the right side. We also discard the $6$-th and $8$-th terms (as their signs are negative). This gives:

$$\Delta \Phi \leq 0.95 + f\left(\frac{a+b+c+2}{c}\right) - f\left(\frac{a+b+c+2}{a}\right).$$

Observe that since $a \geq 1$, $f\left(\frac{a+b+c+2}{c}\right) \leq
f\left(\frac{3a+b+c}{c}\right) =
f\left(1 + \frac{3}{x} + \frac{1}{xy}\right) \leq
1.5 \left( \frac{3}{\gamma^2} + \frac{1}{\gamma} \right) \leq 0.01$, where the next to last inequality follows from Lemma~\ref{lem123}(i). By noting that $f\left(\frac{a+b+c+2}{a}\right) \geq f(1+x) \geq f(\gamma^2) > 1$ we complete the proof in this case. 

{\bf Case $x\geq y \geq \gamma$:} We drop the $3$-rd and $5$-th terms (as a negative quantity). Then it suffices to show:
$$f\left(\frac{b+c+1}{b}\right) \geq f\left(\frac{a+b+1}{a}\right) + f\left(\frac{a+b+1}{b}\right) + f\left(\frac{a+b+c+2}{c}\right).$$
Using Lemma~\ref{lem123}(i) we get:
\begin{align*}
&f \left( \frac{a+b+1}{a} \right) + f\left(\frac{a+b+c+2}{c} \right) \leq
f\left( \frac{a+2b}{a} \right) + f\left(\frac{a+3b+c}{c} \right) \leq \\
&1.5 \cdot \frac{2}{y} + 1.5 \cdot \left( \frac{3}{xy} + \frac{1}{x} \right) \leq \frac{4.5}{y} + \frac{4.5}{y^2}.
\end{align*}

Using Lemma~\ref{lem123}(iii) and by Lagrange theorem:
\begin{align*}
&f\left(\frac{b+c+1}{b}\right) - f\left(\frac{a+b+1}{b}\right) = f\left( 1 + xy + \frac{1}{b} \right) - f\left( 1 + y + \frac{1}{b}\right) \\
&\geq f\left( 1 + \gamma y + \frac{1}{b} \right) - f\left( 1 + y + \frac{1}{b}\right) \geq 
\left[\left( 1 + \gamma y + \frac{1}{b} \right) -  \left( 1 + y + \frac{1}{b}\right)\right] \cdot f'\left( 1 + \gamma y + \frac{1}{b} \right)\\
&\geq\left( \gamma y - y \right) f'\left( 2 + \gamma y \right) 
\geq  \frac{\gamma y - y }{3 (2+\gamma y) \log^2(2+\log(2+\gamma y))}
\geq  \frac{0.3}{ \log^2(2+\log(2+\gamma y))} \geq
\frac{4.5}{y} + \frac{4.5}{y^2}.
\end{align*}

Here, the second inequality uses Lagrange theorem and the concavity of $f$, and the second to last inequality uses that $(\gamma y - y)/(2+\gamma y)\geq 0.9996$.

{\bf Case $y>x$:} We again drop the $3$-rd and $5$-th terms (as a negative quantity). Using Lemma~\ref{lem123}(iii):
\begin{align*}
&f\left(\frac{b+c+1}{b}\right) - f\left(\frac{a+b+1}{b}\right) = f\left( 1 + xy + \frac{1}{b} \right) - f\left( 1 + y + \frac{1}{b}\right) \\
&\geq \left[ \left( 1 + xy + \frac{1}{b} \right)  - \left( 1 + y + \frac{1}{b}\right) \right] \cdot f'\left( 1 + xy + \frac{1}{b} \right)
\geq (xy - y)\cdot f'\left( 2 + xy \right) \\
&\geq \frac{xy-y}{3 (2+xy) \log^2(2+\log(2+xy))}
\geq \frac{0.3}{\log^2(2+\log(2+xy))} 
\geq \frac{0.3}{\log^2(2+\log(2+y^2))}.
\end{align*}


Also, by concavity of $f$:
\begin{align}
f\left(\frac{a+b+c+2}{b+c+1}\right) +
f\left(\frac{b+c+1}{c}\right) \geq
f\left(1 + \frac{a+1}{b+c+1} + \frac{b+1}{c}\right).
\label{ineq26}
\end{align}
Thus, looking at the $4$-th, $6$-th, and $8$-th terms:
\begin{align*}
&f\left(\frac{a+b+c+2}{c}\right) - f\left(\frac{a+b+c+2}{b+c+1}\right) -
f\left(\frac{b+c+1}{c}\right)\\
&\leq f\left(\frac{a+b+c+2}{c}\right) -
f\left(1 + \frac{a+1}{b+c+1} + \frac{b+1}{c}\right)\quad \tag{by Inequality~(\ref{ineq26})} \\
&\leq 1.5 \left(  \frac{a+b+2}{c} - \frac{a+1}{b+c+1} - \frac{b+1}{c} \right) = 1.5 \left(  \frac{a+1}{c} - \frac{a+1}{b+c+1} \right) \quad \tag{by Lemma~\ref{lem123}(i)}\\
&= 1.5 \cdot \frac{(a+1)(b+1)}{c\cdot (b+c+1)}
= 1.5 \cdot \frac{\left(1+\frac{1}{a} \right) \left(1+\frac{1}{b} \right)}{x\cdot \left( 1 + xy + \frac{1}{b}\right)} \leq \frac{1}{xy} \leq \frac{1}{y}. \quad \tag{as $x\geq \gamma^2$, and $\frac{1}{a},\frac{1}{b}\leq 1$}
\end{align*}


Putting it all together (using $f\left(\frac{a+b+1}{a}\right) \leq f\left(1 + \frac{2}{y}\right) \leq \frac{3}{y}$) we get:
$$\Delta\Phi \leq \frac{3}{y} + \frac{1}{y} - \frac{0.3}{\log^2(2+\log(1+y^2))} < 0. $$

\textbf{Type-(3B) link }($c \geq  \gamma b \geq \gamma^2 a$):
The total increase in potential is:
\begin{align*}
\Delta\Phi =f\left(\frac{a+b+1}{a}\right) + f\left(\frac{a+b+1}{b}\right) + f\left(\frac{a+b+c+2}{a+b+1}\right) + f\left(\frac{a+b+c+2}{c}\right) \\
- f\left(\frac{a+b+c+2}{a}\right) - f\left(\frac{a+b+c+2}{b+c+1}\right) - f\left(\frac{b+c+1}{b}\right) - f\left(\frac{b+c+1}{c}\right).
\end{align*}

We drop the $3$-rd and $7$-th terms (as a negative quantity), and also discard the $6$-th term. 

$$\Delta\Phi \leq f\left(\frac{a+b+1}{a}\right) +
f\left(\frac{a+b+1}{b}\right) + f\left(\frac{a+b+c+2}{c}\right) - f\left(\frac{a+b+c+2}{a}\right) - f\left(\frac{b+c+1}{c}\right).$$
Thus, it is enough to show that:

\begin{align}
\label{ineq27}
f\left(\frac{a+b+c+2}{a}\right) + f\left(\frac{b+c+1}{c}\right) \geq f\left(\frac{a+b+1}{a}\right) +
f\left(\frac{a+b+1}{b}\right) + f\left(\frac{a+b+c+2}{c}\right).
\end{align}

Let $x=c/b \geq \gamma,y=b/a \geq \gamma$. We split the rest of the proof into two cases according to the ordering between $x$ and $y$.

{\bf Case $x>y$:} We discard $f\left(\frac{b+c+1}{c}\right)$ in Inequality~(\ref{ineq27}). Note that:
$$f \left( \frac{a+b+1}{b} \right) + f\left(\frac{a+b+c+2}{c}\right) \leq f \left( 1+\frac{2}{y} \right) +  f \left(1+ \frac{1}{x} + \frac{3}{xy} \right) \leq
f \left( 1+\frac{2}{y} \right) +  f \left(1+ \frac{1}{y} + \frac{3}{y^2} \right).$$

Thus, by Lemma~\ref{lem123}(i):

$$f \left( \frac{a+b+1}{b} \right) + f\left(\frac{a+b+c+2}{c}\right) \leq   1.5 \left(  \frac{2}{y} + \frac{1}{y} + \frac{3}{y^2}  \right) = \frac{4.5}{y} + \frac{4.5}{y^2}.$$

Bounding the remaining terms of Inequality~(\ref{ineq27}) we get (using concavity):
\begin{align*}
&f\left(\frac{a+b+c+2}{a}\right) - f\left(\frac{a+b+1}{a}\right)  = f\left(1+y+xy+\frac{2}{a}\right) -
f\left(1+y+\frac{1}{a}\right)\\
&\geq f\left(1+y+y^2+\frac{1}{a}\right) -
f\left(1+y+\frac{1}{a}\right) \mbox{~~~~ ($x>y$)}\\
&\geq \left[ \left(1+y+y^2+\frac{1}{a}\right) -
\left(1+y+\frac{1}{a}\right) \right] \cdot f'\left(1+y+y^2+\frac{1}{a}\right) 
\tag{Lagrange theorem}\\
&= y^2 \cdot f'\left(1+y+y^2+\frac{1}{a}\right) \geq
y^2 \cdot f'(2+y+y^2) \tag{concavity of $f$, $a \geq 1$}\\
&\geq \frac{y^2}{3(2+y+y^2)\log^2(2+\log(2+y+y^2))} 
\geq \frac{0.3}{\log^2(2+\log(2+y+y^2))} \tag{Lemma~\ref{lem123}(iii)}\\ 
&\geq \frac{4.5}{y} + \frac{4.5}{y^2}.
\end{align*}


{\bf Case $x\leq y$:} Taking the $2$-nd, $4$-th, and $5$-th terms from Inequality~(\ref{ineq27}) we get:
\begin{align*}
&f \left( \frac{a+b+1}{b} \right) + f\left(\frac{a+b+c+2}{c}\right) - 
f\left(\frac{b+c+1}{c}\right) \\
&\leq f \left( \frac{2a+b}{b} \right) + f\left(\frac{2a+b+c+1}{c}\right) - 
f\left(\frac{b+c+1}{c}\right) \tag{$a \geq 1$}\\
&= f \left( 1+\frac{2}{y} \right) +  f \left(1+ \frac{1}{x}+ \frac{1}{c} + \frac{2}{xy}  \right)-  f\left(1 + \frac{1}{x} +  \frac{1}{c}\right)\\
&\leq
1.5 \left(\frac{2}{y} + \frac{2}{xy} \right) \leq
1.5 \left(\frac{2}{y} + \frac{2}{\gamma y} \right) =
\frac{3}{y} + \frac{3}{\gamma y}. \tag{Lemma~\ref{lem123}(i)}
\end{align*}

Also, using Lagrange theorem and Lemma~\ref{lem123}(iii) we get:

\begin{align*}
&f\left(\frac{a+b+c+2}{a}\right) - f\left(\frac{a+b+1}{a}\right)  = f\left(1+y+xy+\frac{2}{a}\right) -
f\left(1+y+\frac{1}{a}\right)\\
&\geq f\left(1+y+\gamma y+\frac{1}{a}\right) -
f\left(1+y+\frac{1}{a}\right) \geq
\gamma y \cdot f'\left(1+y+\gamma y+\frac{1}{a}\right) \geq
\gamma y \cdot f'(2+y+\gamma y) \\
&\geq\frac{\gamma y}{3(2+y+\gamma y)\log^2(2+\log(2+y+\gamma y))}
\geq \frac{0.3}{\log^2(2+\log(2+y+\gamma y))} \geq
\frac{3}{y} + \frac{3}{\gamma y},
\end{align*}


finishing the proof.
\end{proof}

\subsection{Proof of Lemma~\ref{cat-type5}}\label{app_lem7}

\begin{proof}
As earlier for type-(1) links, define $x=c/a \geq \gamma^2,y=a/b \geq 1/\gamma$. 

\begin{align*}
- \Delta\Phi & =  f(1 + \frac{y}{1+xy}) + f(1 + 1/y + x) + f(1 + \frac{1}{xy}) + f(1 + xy)\\
& -f(1 + y) - f(1+\frac{1}{y}) - f(1+\frac{1}{x}+\frac{1}{xy}) - f(1 + \frac{xy}{y+1}) - O(1).
\end{align*}

We have $H_A = \log(1+x+1/y) + O(1)$, and $H_B = \log(1+xy)+O(1)$. Note, $H_B \geq \Omega(1) \cdot H_A$.

Collecting some constant terms, we have:
$$- \Delta\Phi \geq  f(1 + xy) + f(1 + x) - f(1+y) - f\left(1 + \frac{xy}{y+1}\right) - O(1).$$

As $f(1 + x) - f(1 + \frac{xy}{y+1}) \geq 0$, we further simplify:
$$- \Delta\Phi \geq  f(1 + xy)  - f(1+y) - O(1).$$

It is now sufficient to show $$f(1 + xy)  - f(1+y) \geq \Omega(1) \cdot \frac{\log{(1+x+1/y)}}{\log^2{\left(2 + \log{(1+xy)}\right)}} - O(1).$$

We start with the left side:
\begin{align*}
f(1 + xy)  - f(1+y) & \geq f(1+xy) - f(2\gamma y)  \\
& = g(\log{(1+xy)}) - g(\log{(y)} + \log{2\gamma})\tag{rewriting}\\
& \geq g'(\log{(1+xy)}) \cdot (\log{(1+xy) - \log{y}} - \log{2\gamma})\tag{by Lagrange thm.}\\
& = g'(\log{(1+xy)}) \cdot (\log{(1/y+x)} - \log{2\gamma})\tag{rewriting}\\
& \geq \Omega(1) \cdot \frac{\log{(1/y + x)}}{\log^2{(2 + \log{(1+xy)})}} - O(1)\tag{expanding the derivative of $g$}\\
& \geq \Omega(1) \cdot \frac{\log{(1 +x+1/y)}}{\log^2{(2 + \log{(1+xy)})}} - O(1)\tag{adding a constant term}\\
& = \Omega(1) \cdot \frac{H_A}{\log^2{(2 + H_B)}} - O(1).
\end{align*}

\end{proof}

\subsection{Proof of Lemma~\ref{cat-type6}}\label{app_lem8}

\begin{proof}
As earlier for type-(1) links, define $x=c/b \geq \gamma,y=b/a \geq \gamma$. We get:
\begin{align*}
- \Delta\Phi & =  f\left(1 + \frac{1}{y+xy}\right) + f(1 + y + xy) + f\left(1 + \frac{1}{x}\right) + f(1 + x)\\
& -f\left(1 + \frac{1}{y}\right) - f(1+y) - f\left(1+\frac{1}{x}+\frac{1}{xy}\right) - f\left(1 + \frac{xy}{y+1}\right) - O(1).
\end{align*}

We have $H_A = \log(1+y+xy) + O(1)$, and $H_B = \log(1+x)+O(1)$. Clearly $H_A \geq H_B$.

Collecting some constant terms, we have:
$$- \Delta\Phi \geq  f(1 + y + xy) + f(1 + x) - f(1+y) - f\left(1 + \frac{xy}{y+1}\right) - O(1).$$

As $f(1 + x) - f(1 + \frac{xy}{y+1}) \geq 0$, we further simplify:
$$- \Delta\Phi \geq  f(1 + y + xy)  - f(1+y) - O(1).$$

It is now sufficient to show $$f(1 + y + xy)  - f(1+y) \geq  \Omega(1) \cdot \frac{\log{(1+x)}}{\log^2{(2 + \log{(1+y+xy)})}} - O(1).$$
We start with the left side:
\begin{align*}
f(1 + y + xy)  - f(1+y) & \geq f(y+xy) - f(1+y) \\
& \geq f(y+xy) - f(2y)  \\
& = g(\log{(y+xy)}) - g(\log{(y)}+1)\tag{rewriting}\\
& \geq g'(\log{(y+xy)}) \cdot (\log{(y+xy) - \log{y} - 1})\tag{by Lagrange thm.}\\
& = g'(\log{(y+xy)}) \cdot (\log{(1+x)}-1)\tag{rewriting}\\
& \geq \Omega(1) \cdot \frac{\log{(1+x)}}{\log^2{(2 + \log{(y+xy)})}} - O(1)\tag{expanding the derivative of $g$}\\
& \geq \Omega(1) \cdot  \frac{\log{(1+x)}}{\log^2{(2 + \log{(1+y+xy)})}} - O(1)\tag{adding a constant term} \\
& = \Omega(1) \cdot \frac{H_B}{\log^2{(2 + H_A)}} - O(1).
\end{align*}

\end{proof}

\newpage
\section{Simpler proof of the $O\left( \log{n}\cdot \log{\log{n}}/\log{\log{\log{n}}} \right)$ bound}\label{section-raman}
Recall that Balasubramanian and Raman~\cite{pathbalance} proved a bound of $O \left(\log{n} \cdot \log\log{n} / \log\log\log{n}\right)$ on the amortized time of operations in path-balanced BSTs, using a scaled version of the ``sum-of-logs'' potential function (see Definition~\ref{sum-of-logs}). We give a somewhat simpler proof of this result by combining an intermediate result of Balasubramanian and Raman with the Fredman et al.\ $O(\log{n} \cdot \log\log{n} / \log\log\log{n})$ amortized time bound for multipass pairing heaps~\cite{pairing}.

As in~\cite{pairing} and~\cite{pathbalance} we use a scaled ``sum-of-logs'' potential function.

\begin{definition}\label{sum-of-logs}
Let $\Psi^{'}$ be sum-of-logs potential function. That is
$$\Psi^{'}(T) = \sum_{x\in T}{\log{s(x)}}.$$
The \emph{scaled} sum-of-logs potential function is denoted by $\Psi$ and equals $\Psi^{'}/\log{\log{\log{n}}}$, where $n$ is the size of the tree.
\end{definition}

We analyse long paths using Theorem~\ref{multipassWarmup} and use Inequality (1) of Balasubramanian and Raman~\cite{pathbalance} to analyse the last recursive call on a short path.

Formally, we prove:

\begin{theorem}\label{thmWarmUp}
The amortized time of search using path-balance is
$O \left( \log{n} \cdot \log\log{n} / \log\log\log{n} \right)$.
\end{theorem}

We state Theorem 2 of Fredman et al.~\cite{pairing} for \emph{multipass pairing heaps} which suffices for us to prove Theorem~\ref{thmWarmUp}. Recall that in Theorem~\ref{main-thm} we improved this result. 

\begin{theorem}[ {{\cite[Thm~2.]{pairing}}} ]\label{multipassWarmup}
The amortized time of delete-min in multipass pairing heaps, using the scaled sum-of-logs potential function is $O \left(\log{n} \cdot \log\log{n} / \log\log\log{n} \right)$.
\end{theorem}

In order to analyse a recursive call on a path longer than a threshold $\tau = \log{n} \cdot \log\log{n} / \log\log\log{n} $ we use the following lemma regarding the potential change caused by rotating a median to the root. 

\begin{lemma}\label{mtrWarmUp}
Rotating the median to the root increases $\Psi^{'}$ by $O(\log{n})$.
\end{lemma}
\begin{proof}
Note that the only nodes that change their potential are the median and the nodes above it on the search path. Among these nodes only the median may increase its potential, as the subtrees of all other nodes cannot increase, see Figure~\ref{rotation}. Therefore, since the potential of each node is bounded by $O(\log{n})$, the result follows.
\end{proof}
In this section we assume that a short path is converted into a balanced binary tree in a single step. We analyse the change in $\Psi$ using Inequality (1) of Balasubramanian and Raman~\cite{pathbalance}.
We restate this inequality here, stressing that we use it only for the warm-up proof of this section and not in the proof of the main results of this paper. 

\begin{lemma}[{{\cite[Ineq.\ (1)]{pathbalance}}}]\label{RammanWeak}
The change in $\Psi^{'}$ caused by converting $\P^x$ into a balanced tree is
$$\Delta \Psi^{'} \leq 3\log{n}\log{(\ell(\P^x)+1)} + 3\log{n} + 1.$$
\end{lemma}

By the definition of $\Psi$ and $\Psi^{'}$, Lemma~\ref{RammanWeak} implies that
$$\Delta \Psi \leq \left( 3\log{n}\log{(\ell(\P^x)+1)} + 3\log{n} + 1 \right) / \log{\log{\log{n}}}.$$
We now prove Theorem~\ref{thmWarmUp}.
\begin{proof}[Proof (Theorem~\ref{thmWarmUp})]
Let $k=\ell(\P^x)$ be the length of the search path ($k$ is also the real cost). Assume first that $ k\leq \tau$. By Lemma~\ref{RammanWeak} we have
$$k + \Delta \Psi \leq T +  \left( 3\log{n}\log{(T+1)} + 3\log{n} + 1 \right) / \log{\log{\log{n}}} =
O\left( \log{n} \cdot \log{\log{n}} / \log{\log{\log{n}}} \right),$$
which is the required \emph{amortized} time.

Assume now $k > \tau$. We look at the potential change caused by the first recursive call (on the original search path of length $k$).

By Lemma~\ref{mtrWarmUp}, rotating the median $m$ to the root causes a potential increase of $O(\log{n})$. The path splits into $\P^{>m}$ and $Q^{x}$, of which $\P^{>m}$ is monotone. By Theorem~\ref{multipassWarmup} and the definition of \emph{amortized} time, performing a multipass transformation on $\P^{>m}$ results in a decrease in $\Psi$ of $k/2 - O \left( \log{n} \cdot \log{\log{n}} / \log{\log{\log{n}}} \right)$. Scaling the potential by $2$, we obtain a decrease of $k - O \left( \log{n} \cdot \log{\log{n}} / \log{\log{\log{n}}} \right)$ in potential. By the same argument, all recursive calls on search paths of length larger than $\tau$ can only decrease the potential.

To conclude, all recursive calls (up until the last call on a path of length less than $\tau$) decrease the potential by at least $k - O \left( \log{n} \cdot \log{\log{n}} / \log{\log{\log{n}}} \right)$. Also, as we have already analysed, the last call on a path of length at most $\tau$ increases $\Psi$ by at most $O \left( \log{n} \cdot \log{\log{n}} / \log{\log{\log{n}}} \right)$, yielding the required amortized time.
\end{proof}

\begin{figure}
    \begin{center}
\includegraphics[width=0.9\textwidth]{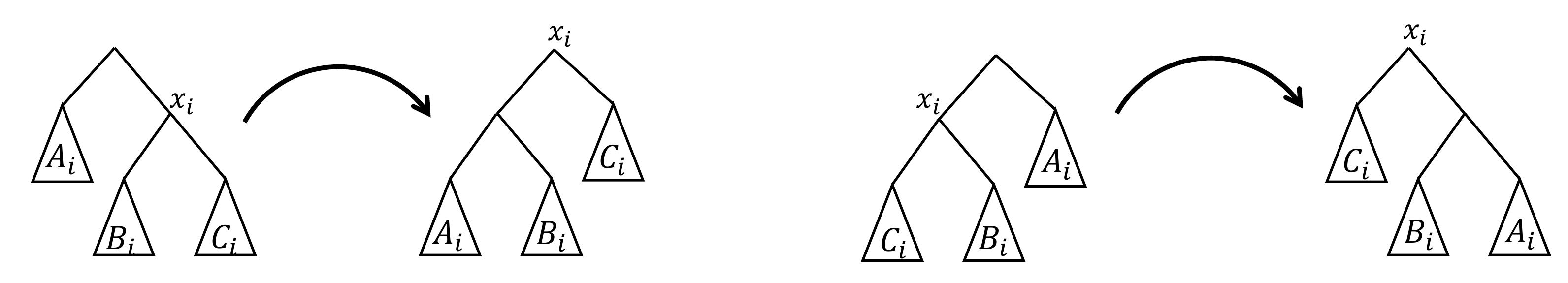}
    \end{center}
\caption{Single rotation of a median $x_i$. Note that only $x_i$ may increase its potential.
\label{rotation}}
\end{figure}

\newpage

\section{Proof of Lemma~\ref{multi_median}}
\label{last_proof}

\begin{figure}[H]
    \begin{center}
\includegraphics[width=0.6\textwidth]{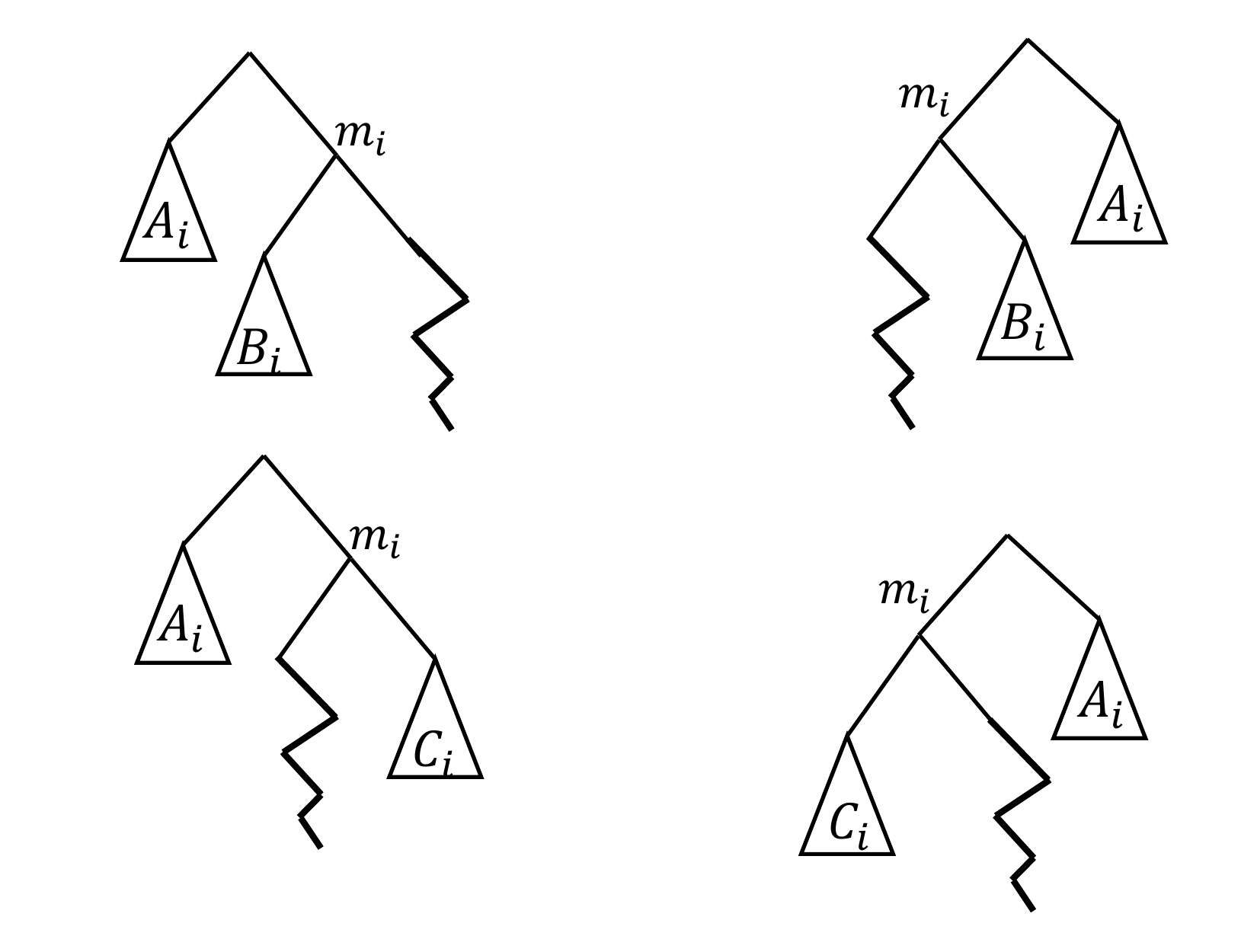}
    \end{center}
\caption{Four cases according to the direction in which the search path continues.
\label{multimedian-cases}}
\end{figure}

\begin{proof}
Let $m_1,...,m_j$ be the medians of all recursive calls (where $j < \log{\log{n}}$). We analyse separately the change in potential due to nodes on $\P$ and nodes off $\P$.

First, we analyse the change in potential due to nodes off $\P$. The only nodes off $\P$ that can increase their potential due to rotating $m_i$ to the root (of the subtree rooted at the shallowest node of $\P^i$) are the children of $m_i$. Each median has at most one such child which is off $\P$ before the transformation. This child is either the root of $B_i$ or the root of $C_i$, see Figure~\ref{multimedian-cases}.
Thus, at most $j < \log{\log{n}}$ nodes off $\P$ increase their potential. Since the potential of each node is bounded by $f(n) = O\left( \log{n} / \log^3{\log{n}} \right)$, we get that the total increase in potential due to nodes off $\P$ is $O(\log{n})$.

We now analyse nodes on $\P$. We look at the potential after the transformation. Let $\Phi_i$ be the potential of the nodes on $\P^{>m_{\ell-i+1}}$ after the transformation. Let $a_1\leq ... \leq a_{\ell (\P^{>m_{\ell -i +1}})}$ be the sizes of the subtrees rooted at the nodes of $\P^{>m_{\ell-i+1}}$ bottom to top, we get that
\begin{align*}
\Phi_i &= \sum^{\ell (\P^{>m_{\ell -i +1}})-1}_{k=1} f \left( \frac{a_{k+1}}{a_k} \right) =
\sum^{\ell (\P^{>m_{\ell -i +1}})-1}_{k=1} g \left( \log{ \frac{a_{k+1}}{a_k} }\right) \leq \ell (\P^{>m_{\ell -i +1}}) \cdot g \left(\frac{\log{n}}{\ell (\P^{>m_{\ell -i +1}})} \right) \\ &= \frac{\log{n}}{\log^3 \left(4 + \frac{\log{n}}{\ell (\P^{>m_{\ell -i +1}})} \right) } \leq
\frac{\log{n}}{\log^3 \left(4 + \frac{\log{n}}{2^i} \right) } \leq
\frac{\log{n}}{\left( \log{\log{n}}  - i   \right)^3}.
\end{align*}

When we sum this up over all of the monotone paths $\P^{>m_i}$ we get $O(\log{n})$.
Thus, the total change in potential caused by the entire transformation is $O(\log{n})$. 
\end{proof}

\newpage

\section{Proof of Theorem~\ref{main-thm} with a modified potential}\label{second-Appendix}
We show that the results of \S\,\ref{sec:mp} also hold with the potential function:
$$\phi(x) = \frac{H(x)}{\log^{3}{(4+H(x))}}.$$

The proofs closely follow the previous ones. We observe that the exact exponent of the potential function (2 or 3) is crucially used only in the proof of Lemma~\ref{lemc} (the remaining proofs go through essentially unchanged with the new potential function). We first update the constant $\gamma$ to $\gamma = 3000^2$.

\begin{lemma}\label{applemc}
Links of type-(3) can only decrease the potential.
\end{lemma}

We reprove the technical lemmas we need in order to prove Lemma~\ref{applemc}.

\begin{lemma}\label{applem1}
For every $x \geq 1,y \geq 0$, $f(x+y)\leq f(x)+1.5y$.
\end{lemma}

\begin{proof}
Due to the concavity of $f$ and the fact that
$f'(1) = \frac{1}{8\ln{2}} < 1.5$. 
\end{proof}

\begin{lemma}\label{applem2}
For every $x\geq \gamma$ it holds that
$$f'(x) \geq \frac{0.32}{x\log^3(4+\log{x})}.$$
\end{lemma}

\begin{proof}
Taking the derivative of $f(x) = \log{x}/\log^3{(4+\log{x})}$ we get by the multiplication rule
\begin{align*}
f'(x) &=
\frac{\ln^2{2}}{x \ln^3(4+\log{x})}
-\frac{\frac{3\ln{2} \ln{x}}{\log{x}+4}}
{x \ln^4(4+\log{x})} \\
&=
\frac{\ln^2{2}}{x\ln^3(4+\log{x})} \left[ 1 - \frac{3\ln{x}}{\ln{2} \cdot (4+\log{x}) \cdot \ln(4+\log{x})} \right] \\
&\geq
\frac{0.224\cdot \ln^2{2}}{x\ln^3(4+\log{x})} =
\frac{0.224}{x \cdot \ln{2} \cdot \log^3(4+\log{x})} \geq
\frac{0.32}{x\log^3(4+\log{x})}.
\end{align*}

\end{proof}

\begin{lemma}\label{applem3}
$f(1+x) + f(1+\frac{1}{x})$ has only one maximum point in
$[\frac{1}{\gamma},\gamma]$, and two global minima.
\end{lemma}

\begin{proof}
Because of symmetry around $x=1$, it suffices to prove that $h(x) = f(1+x) + f(1+\frac{1}{x})$ has only one minimum in $(1,\gamma)$. This minimum is at $x=5.6259$. The plot of this function is similar to the plot in Figure~\ref{wolfram}, we omit the tedious analytical derivation. 
\end{proof}

\begin{lemma}\label{appfsymetrybound}
Fix $a,b \in \mathbb{N}$, if $\frac{1}{\gamma} \leq \frac{a}{b} \leq \gamma$, then
$$f \left( \frac{a+b+1}{a} \right) +
f \left( \frac{a+b+1}{b} \right) \leq 0.22 .$$
\end{lemma}

\begin{proof}
Cases $1 \leq a,b \leq 3000$ are computer-verified (the maximal value is 0.208). Thus, assume $a,b > 3000$. Denote $x = \frac{a}{b}$. Using Lemma~\ref{applem1}, we get that

\begin{align}
\begin{split}
\label{appfsymetry-first}
&f \left( \frac{a+b+1}{a} \right) +
f \left( \frac{a+b+1}{b} \right) =
f \left( 1 + \frac{1}{x} + \frac{1}{a} \right) +
f \left( 1 + x + \frac{1}{b} \right) \\
&\leq f \left( 1 + \frac{1}{x} \right) + f \left( 1 + x \right) + 1.5 \left( \frac{1}{a} + \frac{1}{b} \right) \leq
f \left( 1 + \frac{1}{x} \right) + f \left( 1 + x \right) + 0.001.
\end{split}
\end{align}
Using Lemma~\ref{applem3} we further obtain
\begin{align}
\label{appfsymetry-second}
f \left( 1 + \frac{1}{x} \right) + f \left( 1 + x \right)
\leq \max \{f(2) + f(2), f\left(1+\frac{1}{\gamma}\right) + f(1+\gamma) \} \leq 0.217.
\end{align}
By combining (\ref{appfsymetry-first}) and~(\ref{appfsymetry-second}) the claim follows. 
\end{proof}

We are ready to prove Lemma~\ref{applemc}.

\begin{proof}

\textbf{Type-(3A) link} ($c \geq \gamma^2 a \geq \gamma b$):

Recall that the change in potential given in Equation (\ref{eq:1}) is
\begin{align}
\Delta \Phi = \tag{\ref{eq:1}}
\begin{split}
f\left(\frac{a+b+1}{a}\right) + f\left(\frac{a+b+1}{b}\right) + f\left(\frac{a+b+c+2}{a+b+1}\right) + f\left(\frac{a+b+c+2}{c}\right) \\
- f\left(\frac{a+b+c+2}{a}\right) - f\left(\frac{a+b+c+2}{b+c+1}\right) - f\left(\frac{b+c+1}{b}\right) - f\left(\frac{b+c+1}{c}\right).
\end{split}
\end{align}
Denote $x=\frac{c}{a}\geq \gamma^2, y=\frac{a}{b}\geq \frac{1}{\gamma}$. We split the proof into cases according to the value of $y$.

{\bf Case $y < \gamma$:} We use Lemma~\ref{appfsymetrybound} to bound the first two terms by 0.22. The $7$-th term is larger than the $3$-rd so if we discard them both we only increase the right hand side. We also discard the $6$-th and $8$-th terms (as their signs are negative). This gives

$$\Delta \Phi \leq 0.22 + f\left(\frac{a+b+c+2}{c}\right) - f\left(\frac{a+b+c+2}{a}\right).$$

Observe that since $a \geq 1$, $$f\left(\frac{a+b+c+2}{c}\right) \leq 
f\left(\frac{3a+b+c}{c}\right) =
f\left(1 + \frac{3}{x} + \frac{1}{xy}\right) \leq
1.5 \left( \frac{3}{\gamma^2} + \frac{1}{\gamma} \right) \leq 0.01,$$ where the next to last inequality follows from Lemma~\ref{applem1}. Observing $f\left(\frac{a+b+c+2}{a}\right) \geq f(1+x) \geq f(\gamma^2) > 0.25$ completes the proof in this case. 




{\bf Case $x\geq y \geq \gamma$:} We drop the $3$-rd and $5$-th terms (as a negative quantity). It now suffices to show that
$$f\left(\frac{b+c+1}{b}\right) \geq f\left(\frac{a+b+1}{a}\right) + f\left(\frac{a+b+1}{b}\right) + f\left(\frac{a+b+c+2}{c}\right) .$$
Using Lemma~\ref{applem1} we get:
$$f \left( \frac{a+b+1}{a} \right) + f\left(\frac{a+b+c+2}{c} \right) \leq
1.5 \left( \frac{2}{y} + \frac{1}{x} + \frac{3}{xy} \right) \leq \frac{4.5}{y} + \frac{4.5}{y^2}.$$

Using Lemma~\ref{applem2} and by Lagrange theorem:
\begin{align*}
&f\left(\frac{b+c+1}{b}\right) - f\left(\frac{a+b+1}{b}\right) = f\left( 1 + xy + \frac{1}{b} \right) - f\left( 1 + y + \frac{1}{b}\right) \\
&\geq f\left( 1 + \gamma y + \frac{1}{b} \right) - f\left( 1 + y + \frac{1}{b}\right) \geq 
\left[\left( 1 + \gamma y + \frac{1}{b} \right) -  \left( 1 + y + \frac{1}{b}\right)\right] \cdot f'\left( 1 + \gamma y + \frac{1}{b} \right)\\
&\geq\left( \gamma y - y \right) f'\left( 2 + \gamma y \right) 
\geq  \frac{\left( \gamma y - y \right) \cdot 0.32}{ (2+\gamma y) \log^3(4+\log(2+\gamma y))}
\geq  \frac{0.3}{ \log^3(4+\log(2+\gamma y))} \geq
\frac{4.5}{y} + \frac{4.5}{y^2}.
\end{align*}

Where the second inequality uses Lagrange theorem and the concavity of $f$, and the second to last inequality uses that $(\gamma y - y)/(2+\gamma y)\geq 0.9996$.
%
%
%

{\bf Case $y>x$:} We again drop the $3$rd and $5$th terms (as a negative quantity). Using Lemma~\ref{applem2},
\begin{align*}
&f\left(\frac{b+c+1}{b}\right) - f\left(\frac{a+b+1}{b}\right) = f\left( 1 + xy + \frac{1}{b} \right) - f\left( 1 + y + \frac{1}{b}\right) \\
&\geq \left[ \left( 1 + xy + \frac{1}{b} \right)  - \left( 1 + y + \frac{1}{b}\right) \right] \cdot f'\left( 1 + xy + \frac{1}{b} \right)
\geq (xy - y)\cdot f'\left( 2 + xy \right) \\
&\geq \frac{\left( xy-y\right) \cdot 0.32}{(2+xy) \log^3(4+\log(2+xy))}
\geq \frac{0.3}{\log^3(4+\log(2+xy))} 
\geq \frac{0.3}{\log^3(4+\log(2+y^2))}.
\end{align*}


Also, by the concavity of $f$:
\begin{align}
\label{appeqq-type5-first}
f\left(\frac{a+b+c+2}{b+c+1}\right) +
f\left(\frac{b+c+1}{c}\right) \geq
f\left(1 + \frac{a+1}{b+c+1} + \frac{b+1}{c}\right).
\end{align}
Thus, looking at the $4$-th, $6$-th and $8$-th terms:
\begin{align*}
&f\left(\frac{a+b+c+2}{c}\right) - f\left(\frac{a+b+c+2}{b+c+1}\right) -
f\left(\frac{b+c+1}{c}\right)\\
&\leq f\left(\frac{a+b+c+2}{c}\right) -
f\left(1 + \frac{a+1}{b+c+1} + \frac{b+1}{c}\right) \tag{Inequality (\ref{appeqq-type5-first})} \\
&\leq 1.5 \left(  \frac{a+b+2}{c} - \frac{a+1}{b+c+1} - \frac{b+1}{c} \right) = 1.5 \left(  \frac{a+1}{c} - \frac{a+1}{b+c+1} \right) \tag{Lemma~\ref{applem1}}\\
&= 1.5 \cdot \frac{(a+1)(b+1)}{c\cdot (b+c+1)}
= 1.5 \cdot \frac{\left(1+\frac{1}{a} \right) \left(1+\frac{1}{b} \right)}{x\cdot \left( 1 + xy + \frac{1}{b}\right)} \leq \frac{1}{xy} \leq \frac{1}{y}. \tag{using $x\geq \gamma^2$, and $\frac{1}{a},\frac{1}{b}\leq 1$}
\end{align*}


Putting it all together (using $f\left(\frac{a+b+1}{a}\right) \leq f\left(1 + \frac{2}{y}\right) \leq \frac{3}{y}$) we get that
$$\Delta \Phi \leq \frac{3}{y} + \frac{1}{y} - \frac{0.3}{\log^3(4+\log(1+y^2))} < 0. $$

\textbf{Type-(3B) link} ($c \geq \gamma b \geq \gamma^2 a$):
The total change in potential is
\begin{align}
\Delta \Phi = \tag{\ref{eq:1}}
\begin{split}
f\left(\frac{a+b+1}{a}\right) + f\left(\frac{a+b+1}{b}\right) + f\left(\frac{a+b+c+2}{a+b+1}\right) + f\left(\frac{a+b+c+2}{c}\right) \\
- f\left(\frac{a+b+c+2}{a}\right) - f\left(\frac{a+b+c+2}{b+c+1}\right) - f\left(\frac{b+c+1}{b}\right) - f\left(\frac{b+c+1}{c}\right).
\end{split}
\end{align}
We drop the $3$-rd and $7$-th terms (as a negative quantity), and also discard the $6$-th term:
$$\Delta \Phi \leq f\left(\frac{a+b+1}{a}\right) +
f\left(\frac{a+b+1}{b}\right) + f\left(\frac{a+b+c+2}{c}\right) - f\left(\frac{a+b+c+2}{a}\right) - f\left(\frac{b+c+1}{c}\right).$$
Thus, it is enough to show
\begin{align}
\label{appeqq-type6-first}
f\left(\frac{a+b+c+2}{a}\right) + f\left(\frac{b+c+1}{c}\right) \geq f\left(\frac{a+b+1}{a}\right) +
f\left(\frac{a+b+1}{b}\right) + f\left(\frac{a+b+c+2}{c}\right).
\end{align}

Let $x=\frac{c}{b} \geq \gamma,y=\frac{b}{a} \geq \gamma$. We split the rest of the proof into two cases according to the largest among $x$ and $y$.

{\bf Case $x>y$:} We discard $f\left(\frac{b+c+1}{c}\right)$ from Inequality (\ref{appeqq-type6-first}). Note that
$$f \left( \frac{a+b+1}{b} \right) + f\left(\frac{a+b+c+2}{c}\right) \leq f \left( 1+\frac{2}{y} \right) +  f \left(1+ \frac{1}{x} + \frac{3}{xy} \right) \leq
f \left( 1+\frac{2}{y} \right) +  f \left(1+ \frac{1}{y} + \frac{3}{y^2} \right).$$

Thus, by Lemma~\ref{applem1}

$$f \left( \frac{a+b+1}{b} \right) + f\left(\frac{a+b+c+2}{c}\right) \leq   1.5 \left(  \frac{2}{y} + \frac{1}{y} + \frac{3}{y^2}  \right) = \frac{4.5}{y} + \frac{4.5}{y^2}.$$

Bounding the remaining terms of Inequality (\ref{appeqq-type6-first}), we get (using concavity)
\begin{align*}
&f\left(\frac{a+b+c+2}{a}\right) - f\left(\frac{a+b+1}{a}\right)  = f\left(1+y+xy+\frac{2}{a}\right) -
f\left(1+y+\frac{1}{a}\right)\\
&\geq f\left(1+y+y^2+\frac{1}{a}\right) -
f\left(1+y+\frac{1}{a}\right) \mbox{~~~~ ($x>y$)}\\
&\geq \left[ \left(1+y+y^2+\frac{1}{a}\right) -
\left(1+y+\frac{1}{a}\right) \right] \cdot f'\left(1+y+y^2+\frac{1}{a}\right) 
\tag{Lagrange theorem}\\
&= y^2 \cdot f'\left(1+y+y^2+\frac{1}{a}\right) \geq
y^2 \cdot f'(2+y+y^2) \tag{concavity of $f$, $a \geq 1$}\\
&\geq \frac{y^2\cdot 0.32}{(2+y+y^2)\log^3(4+\log(2+y+y^2))}
\geq \frac{0.3}{\log^3(4+\log(2+y+y^2))} \tag{Lemma~\ref{applem2}}\\
&\geq \frac{4.5}{y} + \frac{4.5}{y^2}.
\end{align*}


{\bf Case $x\leq y$:} By taking the $2$-nd, $4$-th and $5$-th terms from Inequality (\ref{appeqq-type6-first}) we get that
\begin{align*}
&f \left( \frac{a+b+1}{b} \right) + f\left(\frac{a+b+c+2}{c}\right) - 
f\left(\frac{b+c+1}{c}\right) \\
&\leq f \left( \frac{2a+b}{b} \right) + f\left(\frac{2a+b+c+1}{c}\right) - 
f\left(\frac{b+c+1}{c}\right) \tag{$a \geq 1$}\\
&= f \left( 1+\frac{2}{y} \right) +  f \left(1+ \frac{1}{x}+ \frac{1}{c} + \frac{2}{xy}  \right)-  f\left(1 + \frac{1}{x} +  \frac{1}{c}\right)\\
&\leq
1.5 \left(\frac{2}{y} + \frac{2}{xy} \right) \leq
1.5 \left(\frac{2}{y} + \frac{2}{\gamma y} \right) =
\frac{3}{y} + \frac{3}{\gamma y} \tag{Lemma~\ref{applem1}}.
\end{align*}
Also, using Lagrange theorem and Lemma~\ref{applem2} we get that
\begin{align*}
&f\left(\frac{a+b+c+2}{a}\right) - f\left(\frac{a+b+1}{a}\right)  = f\left(1+y+xy+\frac{2}{a}\right) -
f\left(1+y+\frac{1}{a}\right)\\
&\geq f\left(1+y+\gamma y+\frac{1}{a}\right) -
f\left(1+y+\frac{1}{a}\right) \geq
\gamma y \cdot f'\left(1+y+\gamma y+\frac{1}{a}\right) \geq
\gamma y \cdot f'(2+y+\gamma y) \geq \\
&\frac{(\gamma y)\cdot 0.32}{(2+y+\gamma y)\log^3(4+\log(2+y+\gamma y))}
\geq \frac{0.3}{\log^3(4+\log(2+y+\gamma y))} \geq
\frac{3}{y} + \frac{3}{\gamma y}.
\end{align*}


This concludes the proof.
\end{proof}



\newpage

\bibliography{article}


\end{document}